\newif\ifanon\anonfalse
\newcommand{\optappendix}[1]{#1}

\newcommand\blfootnote[1]{%
  \begingroup
  \renewcommand\thefootnote{}\footnote{#1}%
  \addtocounter{footnote}{-1}%
  \endgroup
}

\newcommand{\tikzwrapper}[3]{%
\includegraphics{tikzplots/#1}}

\documentclass[preprint,9pt]{sigplanconf}

\usepackage[font=footnotesize,skip=8pt]{caption}
\usepackage{pgfplots}
\usepackage{tikz}
\usetikzlibrary{plotmarks}
\usetikzlibrary{shapes.symbols}
\usetikzlibrary{calc}
\usetikzlibrary{patterns}
\usetikzlibrary{external}
\tikzsetexternalprefix{tikzplots/}

\usepackage[amssymb]{SIunits}            
\usepackage{courier}            
\usepackage[scaled]{helvet} 
\usepackage{url}                  
\usepackage{listings}          
\usepackage[colorlinks=true,allcolors=blue,breaklinks,draft]{hyperref}   

\usepackage{bbm}
\usepackage{amsmath}
\usepackage{amsthm}
\usepackage{amssymb}
\usepackage{graphicx}
\usepackage{color}
\usepackage{paralist}
\usepackage{caption}
\usepackage{enumitem}      
\usepackage{stmaryrd}
\usepackage{mathpartir}
\usepackage{balance}
\usepackage{titlesec}

\usepackage[caption=false]{subfig}
\usepackage{algorithm}
\usepackage[noend]{algpseudocode}

\usepackage{url}
\usepackage{multirow}
\usepackage{pgfplotstable}
\usepackage{booktabs}

\usepackage[firstpage]{draftwatermark}
\SetWatermarkText{\hspace*{8in}\raisebox{6.75in}{\includegraphics[scale=0.1]{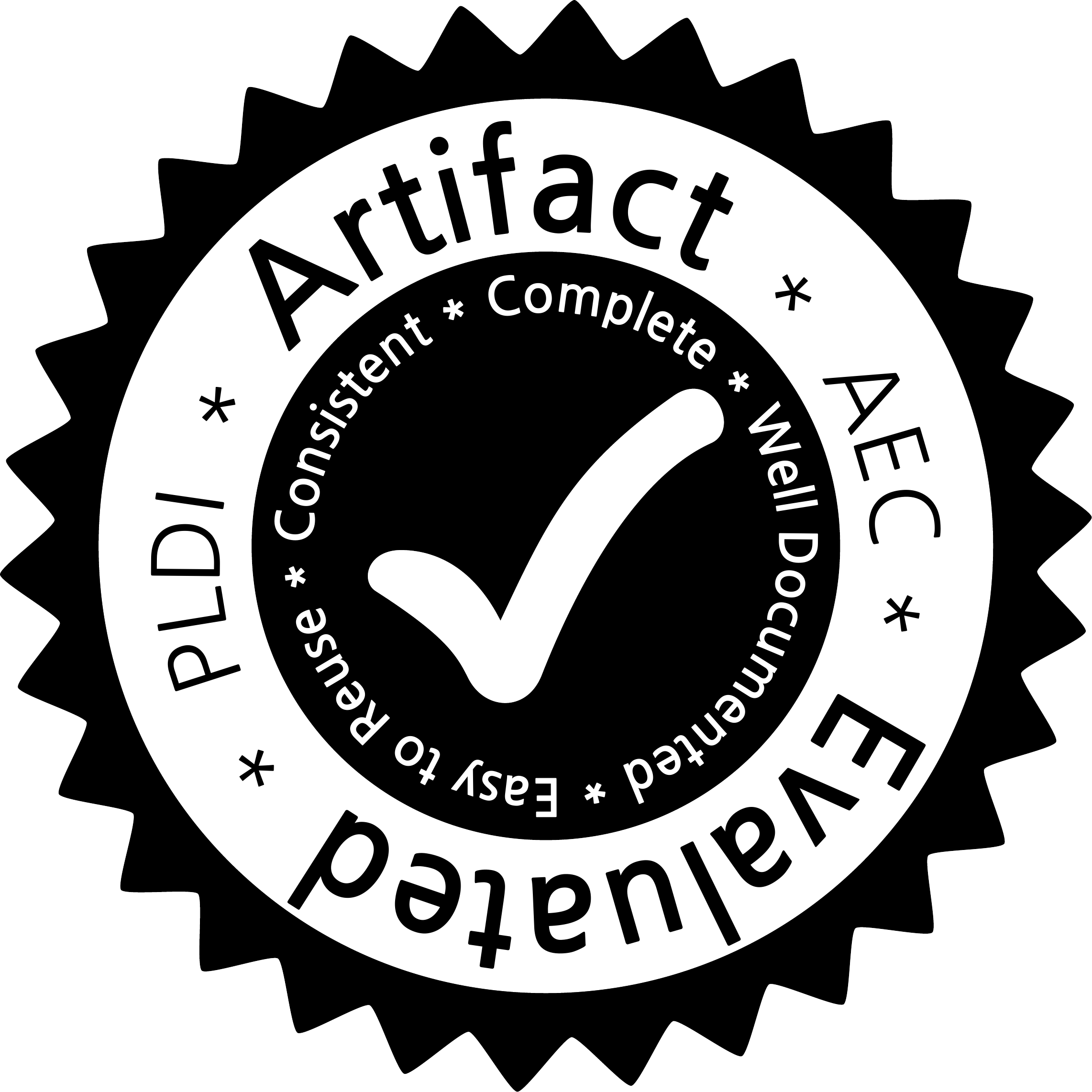}}}
\SetWatermarkAngle{0}

\algnewcommand{\LineComment}[1]{\State \(\triangleright\) #1}

\newcommand\doubleplus{+\kern-1.3ex+\kern0.8ex}

\newcommand{\Holds}{\mathit{Holds}}
\newcommand{\HoldsSink}{\mathit{HoldsSink}}

\newcommand{\labelGraph}{\mathit{labGr}}
\newcommand{\labelNode}{\mathit{labelNode}}
\newcommand{\labeling}{\ensuremath{\lambda}}

\newcommand{\ecl}{\mathit{ecl}}

\newcommand{\succs}{\mathit{succ}}

\newcommand{\AP}{\mathit{AP}}
\newcommand{\traces}{\mathit{traces}}

\newcommand{\follows}{\mathit{follows}}

\newcommand{\HoldsS}{\mathit{Holds_0}}
\newcommand{\relabel}{\mathit{relbl}}
\newcommand{\ancestors}{\mathit{ancestors}}

\newcommand{\comment}[1]{}

\newcommand{\swUpdate}{\ensuremath{\mathit{swUpdate}}}

\newcommand{\NetPolicy}{\ensuremath{\mathit{N}}}

\newcommand{\wait}{\ensuremath{\mathit{wait}}}
\newcommand{\flush}{\ensuremath{\mathit{flush}}}
\newcommand{\epoch}{\ensuremath{\mathit{incr}}}

\newcommand{\Switch}{\ensuremath{s}}

\newcommand{\Epoch}{\ensuremath{\mathit{ep}}}

\newcommand{\Step}[3]{\ensuremath{{#1\overset{#2}{\longrightarrow}#3}}}

\newcommand{\modelCheck}{\ensuremath{\mathit{modelCheck}}}
\newcommand{\incrModelCheck}{\ensuremath{\mathit{incrModelCheck}}}
\newcommand{\checkInitStates}{\ensuremath{\mathit{checkInitStates}}}

\newcommand{\coloneq}{\ensuremath{\mathord{::=}}}

\newcommand{\loc}{\ensuremath{\mathit{loc}}}
\newcommand{\sw}{\ensuremath{\mathit{sw}}}
\newcommand{\pt}{\ensuremath{\mathit{pt}}}

\newcommand{\host}{\ensuremath{\mathit{h}}}

\newcommand{\pkts}{\ensuremath{\mathit{pkts}}}
\newcommand{\pkt}{\ensuremath{\mathit{pkt}}}
\newcommand{\pair}{\ensuremath{\mathit{pr}}}
\newcommand{\pairs}{\ensuremath{\mathit{prs}}}

\newcommand{\cmd}{\ensuremath{\mathit{cmd}}}
\newcommand{\cmds}{\ensuremath{\mathit{cmds}}}
\newcommand{\steps}[1]{\ensuremath{\xrightarrow{#1}}}
\newcommand{\listplus}{\ensuremath{\mathord{@}}}
\newcommand{\onelist}[1]{\ensuremath{[#1]}}
\newcommand{\bagplus}{\ensuremath{\uplus}}
\newcommand{\onebag}[1]{\ensuremath{\{#1\}}}

\newcommand{\denot}[1]{\ensuremath{[\![#1]\!]}}

\newcommand{\pri}{\ensuremath{\mathit{pri}}}
\newcommand{\pat}{\ensuremath{\mathit{pat}}}
\newcommand{\acts}{\ensuremath{\mathit{acts}}}
\newcommand{\act}{\ensuremath{\mathit{act}}}
\newcommand{\rul}{\ensuremath{\mathit{rul}}}
\newcommand{\fld}{\ensuremath{\mathit{f}}}
\newcommand{\ruls}{\ensuremath{\mathit{ruls}}}
\newcommand{\tbl}{\ensuremath{\mathit{tbl}}}
\newcommand{\elts}{\ensuremath{\mathit{Es}}}
\newcommand{\fwd}{\ensuremath{\mathit{fwd}}}
\newcommand{\set}[2]{\ensuremath{#1 \mathord{:=} #2}}

\newcommand{\andalso}{\quad\;}
\newcommand{\Traces}[1]{\ensuremath{\mathcal{T}(#1)}}
\newcommand{\TracesAll}[1]{\ensuremath{\bar{\mathcal{T}}(#1)}}
\newcommand{\Trace}{\ensuremath{\mathcal{T}}}
\newcommand{\Kripke}[1]{\ensuremath{\mathcal{K}(#1)}}
\newcommand{\cons}{\ensuremath{\mathord{::}}}

\definecolor{lightred}{RGB}{255,128,128}
\definecolor{lightblue}{RGB}{128,128,255}
\definecolor{lightgreen}{RGB}{128,255,128}
\definecolor{lightorange}{RGB}{255,204,128}

\definecolor{darkorange}{RGB}{255,132,0}
\definecolor{darkgreen}{RGB}{0,102,0}

\newcommand{\optnewpage}{}

\newcommand{\budget}[1]{}

\newtheorem{theorem}{Theorem}
\newtheorem{lemma}{Lemma}
\newtheorem{corollary}{Corollary}

\newtheorem{definition}{Definition}

\algnewcommand{\Process}[1]{\textbf{Process}\space {\textsc{#1}}}
\algnewcommand{\Assume}[1]{\textbf{assume}\space $(${#1}$)$}
\algnewcommand{\Assert}[1]{\textbf{assert}\space $(${#1}$)$}
\algnewcommand{\LineIf}[2]{\State \algorithmicif\ {#1}\ \algorithmicthen\ {#2}}
\makeatletter
\algnewcommand{\StatexIndent}[1][3]{%
  \setlength\@tempdima{\algorithmicindent}%
  \Statex\hskip\dimexpr#1\@tempdima\relax}
\makeatother

\tikzset{
  >=stealth,
  every loop/.style={min distance=10mm,in=225,out=315,looseness=5},
  every picture/.style={
    draw=black!25,
    very thick,
  },
  switch/.style={
    circle,
    fill=black!40,
    draw=black!50,
    minimum size=10pt,
    inner sep=1pt},
  kripke/.style={
    circle,
    fill=black!40,
    draw=black!50,
    minimum size=12pt,
    inner sep=1pt},
  myblack/.style={
    draw=black!75,
    fill=black!50},
  myempty/.style={
    draw=black!80,
    fill=green!15},
  myempty2/.style={
    draw=black!80,
    fill=blue!15},
  dottedblack/.style={
    dotted,draw=black!75,
    fill=black!50},
  myblue/.style={
    fill=blue!5,
    draw=blue!20 },
  myred/.style={
    fill=black!25!red!60,
    draw=black!50!red!80},
  mycloud/.style={
    draw=white,
  }
}

\algblock{AtomicBegin}{AtomicEnd}
\algnewcommand\algorithmicparfor{\textbf{atomic}}
\algnewcommand\algorithmicendparfor{\textbf{atomic}}
\algrenewtext{AtomicBegin}[1]{\algorithmicparfor\ #1}
\algtext*{AtomicEnd}

\newcommand{\spacehack}[1]{}

\begin{document}

\setlength{\pdfpageheight}{\paperheight}
\setlength{\pdfpagewidth}{\paperwidth}

\conferenceinfo{PLDI'15}{June 13--17, 2015, Portland, OR, USA}
\CopyrightYear{2015}
\crdata{978-1-4503-3468-6/15/06}
\doi{2737924.2737980}

\makeatletter
\def\@ivtitleauthors#1#2#3#4{%
  \if \@andp{\@emptyargp{#2}}{\@emptyargp{#3}}%
    \noindent \@setauthor{40pc}{#1}{\@false}\par
  \else\if \@emptyargp{#3}%
    \noindent \@setauthor{17pc}{#1}{\@false}\hspace{3pc}%
              \@setauthor{17pc}{#2}{\@false}\par
  \else\if \@emptyargp{#4}%
    \noindent \@setauthor{17pc}{#1}{\@false}\hspace{3pc}%
              \@setauthor{17pc}{#3}{\@false}\par
  \else
    \noindent \@setauthor{9.3333pc}{#1}{\@false}\hspace{1.5pc}%
              \@setauthor{9.3333pc}{#2}{\@false}\hspace{1.5pc}%
              \@setauthor{9.3333pc}{#3}{\@false}\hspace{1.5pc}%
              \@setauthor{9.3333pc}{#4}{\@true}\par
    \relax
  \fi\fi\fi
  \vspace{10pt}} 

\def \@maketitle {%
  \begin{center}
  \@settitlebanner
  \let \thanks = \titlenote
  {\leftskip = 0pt plus 0.25\linewidth
   \rightskip = 0pt plus 0.25 \linewidth
   \parfillskip = 0pt
   \spaceskip = .7em
   \noindent \LARGE \bfseries \@titletext \par}
  \vskip 6pt
  \noindent \Large \@subtitletext \par
  \vskip 22pt 
  \ifcase \@authorcount
    \@latex@error{No authors were specified for this paper}{}\or
    \@titleauthors{i}{}{}\or
    \@titleauthors{i}{ii}{}\or
    \@titleauthors{i}{ii}{iii}\or
    \@ivtitleauthors{i}{ii}{iii}{iv}\or
    \@titleauthors{i}{ii}{iii}\@titleauthors{iv}{v}{}\or
    \@titleauthors{i}{ii}{iii}\@titleauthors{iv}{v}{vi}\or
    \@titleauthors{i}{ii}{iii}\@titleauthors{iv}{v}{vi}%
                  \@titleauthors{vii}{}{}\or
    \@titleauthors{i}{ii}{iii}\@titleauthors{iv}{v}{vi}%
                  \@titleauthors{vii}{viii}{}\or
    \@titleauthors{i}{ii}{iii}\@titleauthors{iv}{v}{vi}%
                  \@titleauthors{vii}{viii}{ix}\or
    \@titleauthors{i}{ii}{iii}\@titleauthors{iv}{v}{vi}%
                  \@titleauthors{vii}{viii}{ix}\@titleauthors{x}{}{}\or
    \@titleauthors{i}{ii}{iii}\@titleauthors{iv}{v}{vi}%
                  \@titleauthors{vii}{viii}{ix}\@titleauthors{x}{xi}{}\or
    \@titleauthors{i}{ii}{iii}\@titleauthors{iv}{v}{vi}%
                  \@titleauthors{vii}{viii}{ix}\@titleauthors{x}{xi}{xii}%
  \else
    \@latex@error{Cannot handle more than 12 authors}{}%
  \fi
  \vspace{1.75pc}
  \end{center}}
\makeatother

\def\titletext{Efficient Synthesis of Network Updates}
\title{\titletext}
\preprintfooter{\titletext}


\authorinfo{Jedidiah McClurg}
           {CU Boulder}
           {\small jedidiah.mcclurg@colorado.edu}
\authorinfo{Hossein Hojjat}
           {Cornell University}
           {\small hojjat@cornell.edu}
\authorinfo{Pavol {\v C}ern\'y}
           {CU Boulder}
           {\small pavol.cerny@colorado.edu}
\authorinfo{Nate Foster}
           {Cornell University}
           {\small jnfoster@cs.cornell.edu}

\maketitle

\newcommand{\naive}{na\"{\i}ve}

\begin{abstract}
\noindent Software-defined networking (SDN) is revolutionizing 
the networking industry, but current SDN programming platforms do not
provide automated mechanisms for updating global configurations on the
fly. Implementing updates by hand is challenging for SDN programmers
because networks are distributed systems with hundreds or thousands of
interacting nodes. Even if initial and final configurations are
correct, \naive{}ly updating individual nodes can lead to incorrect
transient behaviors, including loops, black holes, and access control
violations. This paper presents an approach for
automatically synthesizing updates that are guaranteed to preserve
specified properties. We formalize network updates as a distributed
programming problem and develop a synthesis algorithm based on
counterexample-guided search and incremental model checking. 
We describe a prototype
implementation, and present results from experiments on real-world
topologies and properties demonstrating that our tool scales to
updates involving over one-thousand nodes.
\end{abstract}

\category{D.2.4}{Software Engineering}{Software/Program Verification}[Formal methods]
\category{D.2.4}{Software Engineering}{Software/Program Verification}[Model checking]
\category{F.3.1}{Logics and Meanings of Programs}{Specifying and Verifying and Reasoning about Programs}[Logics of programs]
\category{F.4.1}{Mathematical Logic and Formal Languages}{Mathematical Logic}[Temporal logic]
\category{C.2.3}{Computer-communication Networks}{Network Operations}[Network Management]


\keywords
synthesis, verification, model checking, LTL,
network updates, software-defined networking, SDN


\optnewpage
\section{Introduction \budget{2}}

\begin{figure*}[t]
\begin{minipage}{0.5\linewidth}
\centering
\tikzwrapper{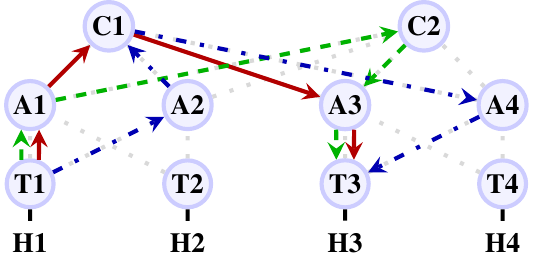}{2in}{
\begin{tikzpicture}[scale=.8]

\coordinate (C-C1) at (2,0);
\coordinate (C-C2) at (6,0);

\coordinate (C-A1) at (1,-1);
\coordinate (C-A2) at (3,-1);
\coordinate (C-A3) at (5,-1);
\coordinate (C-A4) at (7,-1);

\coordinate (C-T5) at (1,-2);
\coordinate (C-T6) at (3,-2);
\coordinate (C-T7) at (5,-2);
\coordinate (C-T8) at (7,-2);

\coordinate (C-H1) at (1,-2.75);
\coordinate (C-H2) at (3,-2.75);
\coordinate (C-H3) at (5,-2.75);
\coordinate (C-H4) at (7,-2.75);

\coordinate (C-S3) at (1,-1);
\coordinate (C-S4) at (2,2);
\coordinate (C-S5) at (2,0);
\coordinate (C-S6) at (2,-2);

\coordinate (C-CW) at (-1.7,0);
\coordinate (C-CC) at (4,0);
\coordinate (C-D) at (3.4,-2);

\node[switch,myblue] (S-C1) at (C-C1) {{\textbf{\footnotesize C1}}};
\node[switch,myblue] (S-C2) at (C-C2) {{\textbf{\footnotesize C2}}};

\node[switch,myblue] (S-A1) at (C-A1) {{\textbf{\footnotesize A1}}};
\node[switch,myblue] (S-A2) at (C-A2) {{\textbf{\footnotesize A2}}};
\node[switch,myblue] (S-A3) at (C-A3) {{\textbf{\footnotesize A3}}};
\node[switch,myblue] (S-A4) at (C-A4) {{\textbf{\footnotesize A4}}};

\node[switch,myblue] (S-T5) at (C-T5) {{\textbf{\footnotesize T1}}};
\node[switch,myblue] (S-T6) at (C-T6) {{\textbf{\footnotesize T2}}};
\node[switch,myblue] (S-T7) at (C-T7) {{\textbf{\footnotesize T3}}};
\node[switch,myblue] (S-T8) at (C-T8) {{\textbf{\footnotesize T4}}};

\node[mycloud] (S-H1) at (C-H1) {{\textbf{\footnotesize H1}}};
\node[mycloud] (S-H2) at (C-H2) {{\textbf{\footnotesize H2}}};
\node[mycloud] (S-H3) at (C-H3) {{\textbf{\footnotesize H3}}};
\node[mycloud] (S-H4) at (C-H4) {{\textbf{\footnotesize H4}}};

\draw[black!15,loosely dotted] (S-A1) edge (S-C1);
\draw[black!15,loosely dotted] (S-A2) edge (S-C1);
\draw[black!15,loosely dotted] (S-A3) edge (S-C1);
\draw[black!15,loosely dotted] (S-A4) edge (S-C1);

\draw[black!15,loosely dotted] (S-A1) edge (S-C2);
\draw[black!15,loosely dotted] (S-A2) edge (S-C2);
\draw[black!15,loosely dotted] (S-A3) edge (S-C2);
\draw[black!15,loosely dotted] (S-A4) edge (S-C2);

\draw[black!15,loosely dotted] (S-T5) edge (S-A1);
\draw[black!15,loosely dotted] (S-T6) edge (S-A2);
\draw[black!15,loosely dotted] (S-T7) edge (S-A3);
\draw[black!15,loosely dotted] (S-T8) edge (S-A4);

\draw[black!15,loosely dotted] (S-T5) edge (S-A2);
\draw[black!15,loosely dotted] (S-T6) edge (S-A1);
\draw[black!15,loosely dotted] (S-T7) edge (S-A4);
\draw[black!15,loosely dotted] (S-T8) edge (S-A3);

\draw[black] (S-H1) edge (S-T5);
\draw[black] (S-H2) edge (S-T6);
\draw[black] (S-H3) edge (S-T7);
\draw[black] (S-H4) edge (S-T8);

\draw[black!30!red] (S-T5.70) edge[->] (S-A1.290);
\draw[black!30!red] (S-A1) edge[->] (S-C1);
\draw[black!30!red] (S-C1) edge[->] (S-A3);
\draw[black!30!red] (S-A3.290) edge[->] (S-T7.70);

\draw[black!30!green,dashed] (S-T5.110) edge[->] (S-A1.250);
\draw[black!30!green,dashed] (S-A1) edge[->] (S-C2);
\draw[black!30!green,dashed] (S-C2) edge[->] (S-A3);
\draw[black!30!green,dashed] (S-A3.250) edge[->] (S-T7.110);

\draw[black!30!blue,loosely dashdotted] (S-T5) edge[->] (S-A2);
\draw[black!30!blue,loosely dashdotted] (S-A2) edge[->] (S-C1);
\draw[black!30!blue,loosely dashdotted] (S-C1) edge[->] (S-A4);
\draw[black!30!blue,loosely dashdotted] (S-A4) edge[->] (S-T7);
\end{tikzpicture}}
\spacehack{-.5em}
\caption{Example topology.}
\spacehack{-.5em}
\label{fig:fattree}
\end{minipage}
\begin{minipage}{0.5\linewidth}
\centering
\tikzwrapper{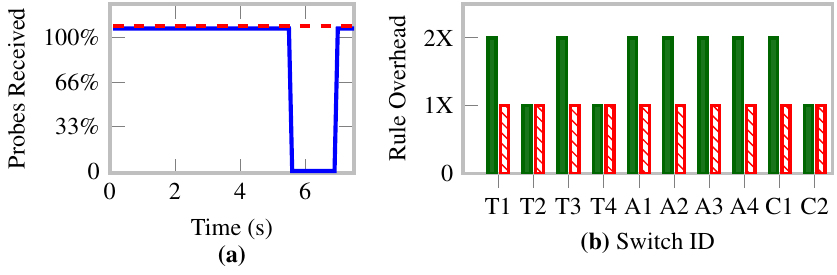}{3.25in}{
\begin{tikzpicture}\scriptsize
\begin{axis}[
    name=mininet,
    xlabel style={align=center},
    ylabel style={align=center},
    xlabel={Time (s)\\{\bf(a)}},
    ylabel={Probes Received},
    xticklabels={X,0,2,4,6},
    yticklabels={x,0,33\%,66\%,100\%},
    legend pos = north west,
    width=1.6in, 
    height=1.3in,
    xmin=0,
    ymin=-1,
    ymax=75,
    xmax=75,
    legend style={fill=none,font=\tiny,line width=0.5pt,row sep=-1ex},
    legend cell align=left]
\addplot [mark=circle,blue,mark size=1pt] table[x index=0,y index=1, header=true,col sep=comma] {../../experiments/data/mininet_tpu.csv};
\addplot [mark=circle,red,dashed,mark size=1pt] table[x index=0,y index=1, header=true,col sep=comma] {../../experiments/data/mininet_syn.csv};
\end{axis}
\begin{axis}[
    at={($(mininet.south east)+(4.5em,0em)$)},
    ybar=.4ex,
    bar width=.75ex,
    name=mininet_overhead,
    xlabel style={align=center},
    ylabel style={align=center},
    ylabel={Rule Overhead},
    xlabel={{\bf(b)} Switch ID},
    xtick pos=left,
    ytick pos=left,
    xtick=data,
    xticklabels={T1,T2,T3,T4,A1,A2,A3,A4,C1,C2},
    yticklabels={0,0,1X,2X},
    legend pos = north west,
    width=2.1in,
    height=1.3in,
    xmin=0,
    xmax=10.5,
    ymin=0,
    ymax=2.5,
    line width=1pt,
    legend style={fill=none,font=\tiny,line width=0.5pt,row sep=-1ex},
    legend cell align=left]
\addplot [mark=circle,darkgreen,fill=darkgreen!90,mark size=1pt] table[x index=0,y index=1, header=true,col sep=comma] {../../experiments/data/mininet_overhead_tpu.csv};
\addplot [mark=circle,red,pattern=north west lines,pattern color=red,mark size=1pt] table[x index=0,y index=1, header=true,col sep=comma] {../../experiments/data/mininet_overhead_syn.csv};
\end{axis}
\end{tikzpicture}}
\spacehack{-2em}
\vspace{-1em}
\caption{Example \naive~(blue/solid-line), two-phase (green/solid-bar), and ordering (red/dashed) updates: (a) probes received; (b) per-switch rule overhead.}
\label{fig:mininet}
\spacehack{-.5em} 
\end{minipage}
\end{figure*}

Software-defined networking (SDN) is a new paradigm in which a
logically-centralized controller manages a collection
of programmable switches. The controller responds to events
such as topology changes, shifts in traffic load, or new connections
from hosts, by pushing forwarding rules to the switches, which
process packets efficiently using specialized hardware. Because the
controller has global visibility and full control over the entire
network, SDN makes it possible to implement a wide variety of network
applications ranging from basic routing to traffic
engineering, datacenter virtualization, fine-grained access control,
etc.~\cite{cacm-sdn-abstractions}. SDN has been used in
production enterprise, datacenter, and wide-area networks, and new
deployments are rapidly emerging.

Much of SDN's power stems from the controller's ability to change
the \emph{global} state of the network. Controllers can
set up end-to-end forwarding paths, provision bandwidth to optimize
utilization, or distribute access control rules to defend against
attacks. However, implementing these global changes in a running
network is not easy. Networks are complex systems with many
distributed switches, but the controller can only modify the
configuration of one switch at a time. Hence, to implement a global
change, an SDN programmer must explicitly transition the network
through a sequence of intermediate configurations to reach the
intended final configuration. The code needed to implement this
transition is tedious to write and prone to error---in general, the
intermediate configurations may exhibit new behaviors that would not
arise in the initial and final configurations.

Problems related to network updates are not unique to SDN. Traditional
distributed routing protocols also suffer from anomalies during periods of
reconvergence, including transient forwarding loops, blackholes, and
access control violations. For users, these anomalies manifest themselves as service outages, degraded performance, and broken
connections. The research community has developed techniques for
preserving certain invariants during
updates~\cite{francois2007avoidingTransient,raza2011graceful,vanbever2011seamless},
but none of them fully solves the problem, as they are limited to
specific protocols and properties. For example, \emph{consensus
routing} uses distributed snapshots to ensure connectivity, but only
applies to the Border Gateway Protocol (BGP)~\cite{john2008consensus}.

It might seem that SDN would exacerbate update-related problems by
making networks even more dynamic---in particular, most current
platforms lack mechanisms for implementing updates in a graceful
way. However, SDN offers opportunities to develop high-level
abstractions for implementing updates automatically while preserving
key invariants. The authors of B4---the controller managing Google's
world-wide inter-datacenter network---describe a vision where:
``multiple, sequenced manual operations [are] not involved [in]
virtually any management operation''~\cite{Google-B4}.

Previous work proposed the notion of a \emph{consistent
update}~\cite{reitblatt2012abstractions}, which ensures that every
packet is processed either using the initial configuration or the
final configuration but not a mixture of the two. Consistency is a
powerful guarantee preserving \emph{all} safety properties, but it is
expensive.  The only general consistent update mechanism is {\em
two-phase update}, which tags packets with versions and maintains
rules for the initial/final configurations simultaneously. This leads
to problems on switches with limited memory and can also make update
time slower due to the high degree of rule churn.

We propose an alternative. Instead of forcing SDN operators to
implement updates by hand (as is typically done today), or using
powerful but expensive mechanisms like two-phase update, we develop
an approach for synthesizing correct update programs efficiently and
automatically from formal specifications. Given initial and final
configurations and a Linear Temporal Logic (LTL) property capturing
desired invariants during the update, we either generate an SDN
program that implements the initial-to-final transition while ensuring
that the property is never violated, or fail if no such program
exists. Importantly, because the synthesized program is only required
to preserve the specified properties, it can leverage strategies that
would be ruled out in other approaches. For example, if the programmer
specifies a trivial property, the system can update switches in any
order. However, if she specifies a more complex property
(e.g. firewall traversal) then the space of possible updates is more
constrained. In practice, our synthesized programs require less memory
and communication than competing approaches.

Programming updates correctly is challenging due to the concurrency
inherent in networks---switches may interleave packet and control
message processing arbitrarily. Hence, programmers must carefully
consider all possible event orderings, inserting synchronization
primitives as needed. Our algorithm works by searching through the
space of possible sequences of individual switch updates, learning
from counterexamples and employing an incremental model checker to
re-use previously computed results.  Our model checker is {\em
incremental} in the sense that it exploits the loop-freedom of correct
network configurations to enable efficient re-checking of properties
when the model changes.
Because the synthesis algorithm
poses a series of closely-related model checking questions, the
incrementality yields enormous performance gains
on real-world update scenarios.

We have implemented the algorithm and heuristics to further
speed up synthesis and eliminate spurious synchronization. We have
interfaced the tool with Frenetic~\cite{frenetic}, synthesized updates
for OpenFlow switches, and used our system to process actual traffic
generated by end-hosts. We ran experiments on a suite of real-world
topologies, configurations, and properties---our results demonstrate
the effectiveness of synthesis, which scales to over one-thousand switches,
and incremental model checking, which outperforms a popular symbolic
model checker used in {\it batch} mode, and a state-of-the-art network
model checker used in incremental mode.

In summary, the main contributions of this paper are:
\setlength{\pltopsep}{0.25em}
\setlength{\plitemsep}{0.35em}
\begin{compactitem}
\item We investigate using synthesis to automatically generate
  network updates (\S \ref{overview}).
\item We develop a simple operational model of SDN and formalize 
  the network update problem precisely (\S \ref{model}).
\item We design a counterexample-guided search algorithm that solves  
  instances of the network update problem, and prove this algorithm to
  be correct (\S \ref{synthesis}).
\item We present an incremental LTL model checker for loop-free models
  (\S \ref{checking}).
\item We describe an OCaml implementation with backends to third-party model 
  checkers and conduct experiments on real-world networks and
  properties, demonstrating strong performance improvements
  (\S \ref{evaluation}).~$\dagger$
\end{compactitem}
Overall, our work takes a challenging network programming problem and
automates it, yielding a powerful tool for building dynamic SDN
applications that ensures correct, predictable, and efficient
network behavior during updates.

\section{Overview \budget{1}}
\label{overview}
\label{ex:redgreen}

\blfootnote{$\dagger$~The PLDI 2015 Artifact Evaluation Committee (AEC) found that our tool
``met or exceeded expectations." }%
To illustrate key challenges related to network updates, consider the
network in Figure~\ref{fig:fattree}. It represents a simplified
datacenter topology~\cite{alfares08} with core switches (C1 and C2),
aggregation switches (A1 to A4), top-of-rack switches (T1 to T4), and
hosts (H1 to H4). Initially, we configure switches to forward traffic
from H1 to H3 along the solid/red path: T1-A1-C1-A3-T3. Later, we
wish to shift traffic from the red path to the dashed/green path,
T1-A1-C2-A3-T3 (perhaps to take C1 down for maintenance). To implement
this update, the operator must modify forwarding rules on switches A1
and C2, but note that certain update sequences break
connectivity---e.g., updating A1 followed by C2 causes packets to be
forwarded to C2 before it is ready to handle
them. Figure~\ref{fig:mininet}(a) demonstrates this with a simple
experiment performed using our system. Using the Mininet
network simulator and OpenFlow switches, we continuously sent ICMP
(\texttt{ping}) probes during a ``\naive'' update (blue/solid line) and
the ordering update synthesized by our tool (red/dashed line). With
the \naive~update, 100\% of the probes are lost during an interval,
while the ordering update maintains connectivity.

\paragraph*{Consistency.}
Previous work \cite{reitblatt2012abstractions} introduced the notion
of a {\em consistent update} and also developed general mechanisms for
ensuring consistency. An update is said to be consistent if every
packet is processed entirely using the initial configuration or
entirely using the final configuration, but never a mixture of the
two. For example, updating A1 followed by C2 is not consistent because
packets from H1 to H3 might be dropped instead of following the red
path or the green path. One might wonder whether preserving
consistency during updates is important, as long as the network
eventually reaches the intended configuration, since most networks only
provide best-effort packet delivery. While it is true that errors can
be masked by protocols such as TCP when packets are lost, there is
growing interest in strong guarantees about network behavior. For
example, consider a business using a firewall to protect internal
servers, and suppose that they decide to migrate their infrastructure to a
virtualized environment like Amazon EC2. To ensure that this new
deployment is secure, the business would want to maintain the same
isolation properties enforced in their home office. However, a
best-effort migration strategy that only eventually reaches the target
configuration could step through arbitrary intermediate states, some
of which may violate this property.

\paragraph*{Two-Phase Updates.} 
Previous work introduced a general consistency-preserving technique
called \emph{two-phase update}~\cite{reitblatt2012abstractions}.  The
idea is to explicitly tag packets upon ingress and use these version
tags to determine which forwarding rules to use at each
hop. Unfortunately, this has a significant cost. During the
transition, switches must maintain forwarding rules for \emph{both}
configurations, effectively doubling the memory requirements needed to
complete the update. This is not always practical in networks where
the switches store forwarding rules using ternary content-addressable
memories (TCAM), which are expensive and
power-hungry. Figure~\ref{fig:mininet}(b) shows the results of
another simple experiment where we measured the total number of rules
on each switch: with two-phase updates, several switches have twice
the number of rules compared to the synthesized ordering update. Even
worse, it takes a non-trivial amount of time to modify forwarding
rules---sometimes on the order of 10ms per rule~\cite{dionysus14}!
Hence, because two-phase updates modify a large number of rules, they
can increase update latency. These overheads can make
two-phase updates a non-starter.

\paragraph*{Ordering Updates.}
Our approach is based on the observation that consistent (two-phase)
updates are overkill in many settings. Sometimes consistency can be
achieved by simply choosing a correct order of switch updates. We call
this type of update an {\em ordering update}. For example, to update
from the red path to the green path, we can update C2 followed by
A1. Moreover, even when we cannot achieve full consistency, we can
often still obtain sufficiently strong guarantees for a specific
application by carefully updating the switches in a particular
order. To illustrate, suppose that instead of shifting traffic to the
green path, we wish to use the blue (dashed-and-dotted) path:
T1-A2-C1-A4-T3. It is impossible to transition from the red path to
the blue path by ordering switch updates without breaking consistency:
we can update A2 and A4 first, as they are unreachable in the initial
configuration, but if we update T1 followed by C1, then packets can
traverse the path T1-A2-C1-A3-T3, while if we update C1 followed by
T1, then packets can traverse the path T1-A1-C1-A4-T3. Neither of
these alternatives is allowed in a consistent update. This failure to
find a consistent update hints at a solution: if we only care about
preserving connectivity between H1 and H3, then either path is
actually acceptable. Thus, either updating C1 before T1, or T1 before
C1 would work. Hence, if we relax strict consistency and instead
provide programmers with a way to specify properties that must be
preserved across an update, then ordering updates will exist in many
situations. Recent work
\cite{wattenhoferconsistent,dionysus14} has explored ordering updates,
but only for specific properties like loop-freedom, blackhole-freedom,
drop-freedom, etc. %
Rather than handling a fixed set of ``canned'' properties, we use a
specification language that is expressive enough to encode these
properties and others, as well as conjunctions/disjunctions of
properties---e.g.  enforcing loop-freedom {\em and} service-chaining
during an update.

\paragraph*{In-flight Packets and Waits.}
Sometimes an additional synchronization primitive is needed to
generate correct ordering updates (or correct two-phase updates, for
that matter). Suppose we want to again transition from the red path to
blue one, but in addition to preserving connectivity, we want every packet
to traverse either A2 or A3 (this scenario might arise if those
switches are actually middleboxes which scrub malicious packets before
forwarding).  Now consider an update that modifies the configurations
on A2, A4, T1, C1, in that order. Between the time that we update T1
and C1, there might be some packets that are forwarded by T1 before it
is updated, and are forwarded by C1 after it is updated. These packets
would not traverse A2 or A3, and so indicate a violation of the
specification. To fix this, we can simply pause after updating T1
until any packets it previously forwarded have left the network. We
thus need a command ``wait'' that pauses the controller for a
sufficient period of time to ensure that in-flight packets have exited
the network. Hence, the correct update sequence for this example would
be as above, with a ``wait'' between T1 and C1. Note that two-phase
updates also need to wait, once per update, since we must ensure that
all in-flight packets have left the network before deleting the old
version of the rules on switches.  Other approaches have traded off
control-plane waiting for stronger consistency, e.g.
\cite{LRFS14} performs updates in ``rounds" that are analogous to
``wait" commands, and Consensus Routing~\cite{john2008consensus}
relies on timers to obtain wait-like functionality.  Note that the
single-switch update time can be on the order of
seconds~\cite{dionysus14,tango}, whereas typical datacenter transit
time (the time for a packet to traverse the network) is much lower,
even on the order of microseconds~\cite{dctcp}. Hence, waiting for
in-flight packets has a negligible overall effect. In addition, our
reachability-based heuristic eliminates most waits in practice.

\paragraph*{Summary.}
This paper presents a sound and complete algorithm and 
implementation for synthesizing a large class of ordering updates
efficiently and automatically. The updates we generate initially modify each switch at
most once and ``wait'' between updates to switches, but a heuristic
removes an overwhelming majority of unnecessary waits in practice. For
example, in switching from the red path to the blue path (while
preserving connectivity from H1 to H3, and making sure that each
packet visits either A3 or A4), our tool produces the following
sequence: update A2, then A4, then T1, then wait, then update C1. The
resulting update can be executed using the Frenetic SDN platform and
used with OpenFlow switches---e.g., we generated
Figure~\ref{fig:mininet}~(a-b) using our tool.

\optnewpage
\section{Preliminaries and Network Model \budget{2}}
\label{model}

\begin{figure*}[t]
\footnotesize
\renewcommand{\dots}{\ensuremath{..}}
\fbox{~\begin{minipage}{0.975\textwidth}
\(
\begin{array}{@{~}l@{~~}|@{~~}l@{~~}|@{~~}l@{~}}
\begin{array}{llcl}
\textit{Switch} & \sw  & \in & \mathbb{N}\\
\textit{Port} & \pt  & \in & \mathbb{N}\\
\textit{Host} & \host  & \in & \mathbb{N}\\
\textit{Priority} & \pri & \in & \mathbb{N}\\
\textit{Epoch} & \Epoch & \in & \mathbb{N}\\
\textit{Field} & \fld & \coloneq & \mathit{src} \mid \mathit{dst} \mid \mathit{typ} \mid \dots\\
\end{array} & \begin{array}{llcl}
\textit{Packet} & \pkt & \coloneq & \{ \fld_1; \dots; \fld_k \}\\
\textit{Pair} & \pair & \coloneq & (\pkt,\pt)\\
\textit{Pattern} & \pat & \coloneq & \{ \pt?; \fld_1?; \dots; \fld_k? \}\\
\textit{Action} & \act & \coloneq & \fwd~\pt \mid \set{\fld}{n} \\
\textit{Rule} & \rul & \coloneq & \{ \pri; \pat; \acts \}\\
\textit{Table} & \tbl & \coloneq & \ruls\\
\end{array} & \begin{array}{llcl}
\textit{Location} & \loc & \coloneq & \host \mid (\sw,\pt) \\
\textit{Command} & \cmd & \coloneq & (\sw, \tbl) \mid \epoch \mid \flush \\
\textit{Switch} & S & \coloneq & \{ \sw; \tbl; \pairs \}\\
\textit{Link} & L & \coloneq & \{ \loc; \pkts; \loc' \}\\
\textit{Controller} & C & \coloneq & \{ \cmds; \Epoch \}\\
\textit{Element} & E & \coloneq & S \mid L \mid C\\
\end{array}
\end{array}
\)

\smallskip
\hrule
\smallskip
\textbf{Data Plane}
\vspace*{-.75em}
\begin{mathpar}
\hspace*{-1pc}
\inferrule*[Right=In]{ 
  L.\loc = \host \andalso L.\loc' = (\sw', \pt') \andalso L.\pkts
  = \pkts \andalso
  C.\Epoch = \Epoch 
}{ 
  C, \; L
  \steps{} 
  C, \; 
  \{ L~\text{with}~\pkts = \pkt^\Epoch \cons \pkts \}
}
\and
\inferrule*[Right=Out]{ 
  L.\loc = (\sw, \pt) \andalso
  L.\loc' = \host \andalso
  L.\pkts = (\pkt^\Epoch \cons \pkts)
}{ 
  L
  \steps{(\sw,\pt,\pkt)}
  \{ L~\text{with}~\pkts = \pkts \}
}
\vspace{-0.5em} \\
\inferrule*[Right=Process]{ 
  L.loc' = (\sw,\pt) \andalso
  L.\pkts = (\pkt^\Epoch \cons \pkts) \andalso
  S.\sw = \sw \andalso
  \denot{S.\tbl}(\pkt,\pt) = \{ (\pkt_1,\pt_1),\dots,(\pkt_n,\pt_n) \} 
}{ 
  L, S
  \steps{(\sw,\pt,\pkt)} 
  \{ L~\text{with}~\pkts = \pkts \},
  \{ S~\text{with}~\pairs = S.\pairs \bagplus \{(\pkt_1^\Epoch,\pt_1),\dots,(\pkt_n^\Epoch,\pt_n)\} \}
}
\vspace{-0.5em} \\
\inferrule*[Right=Forward]{ 
  S.\sw = \sw \andalso 
  S.\pairs = \onebag{(\pkt^\Epoch,\pt)} \bagplus \pairs \andalso
  L.\loc = (\sw,\pt) \andalso
}{
  S,\; L
  \steps{}
  \{ S~\text{with}~ \pairs = \pairs \},\;
  \{ L~\text{with}~ \pkts = L.\pkts @ [\pkt^\Epoch] \} 
}   
\end{mathpar}

\hrule
\smallskip
\textbf{Control Plane and Abstract Machine}
\vspace*{-1em}

\begin{mathpar}
\inferrule*[Right=Update]{ 
  C.\cmds = ((\sw,\tbl) \cons \cmds) \andalso
  S.\sw = \sw 
}{ 
  C,\; S
  \steps{}
  \{ C~\text{with}~ \cmds = \cmds \},\;
  \{ S~\text{with}~\tbl = \tbl \}
}
\and
\inferrule*[Right=Incr]{ 
  C.\cmds = (\epoch \cons \cmds)
}{ 
  C
  \steps{}
  \{ C~\text{with}~ \cmds = \cmds; \Epoch = C.\Epoch + 1 \} 
}
\vspace{-0.5em} \\
\hspace*{-4pc}
\inferrule*[Right=Flush]{ 
  C.\cmds = (\flush \cons \cmds) \andalso
  \Epoch(S_1,\dots,S_k,\;L_1,\dots,L_m) = C.\Epoch
}{ 
  S_1,\dots,S_k,\;
  L_1,\dots,L_m,\;
  C
  \steps{}
  S_1,\dots,S_k,\;
  L_1,\dots,L_m,\;
  \{ C~\text{with}~\cmds = \cmds \}
}
\and
\inferrule*[Right=Congruence]
{ \elts_1 \steps{o} \elts_1' }
{ \elts_1 \uplus \elts_2 \steps{o} \elts_1' \uplus \elts_2 }
\end{mathpar}

\end{minipage}~}
\caption{Network model.}
\label{fig:model}
\spacehack{-.75em}
\end{figure*}

To facilitate precise reasoning about networks during updates, we
develop a formal model in the style of Chemical Abstract
Machine~\cite{berry+:chemical-machine}. This model captures key
network features using a simple operational semantics. It is
similar to the one used by \cite{verified-pldi13}, but is streamlined
to model features most relevant to updates.

\subsection{Network Model}

\paragraph*{Basic structures.}
Each \emph{switch} $\sw$, \emph{port} $\pt$, or \emph{host} $\host$ is
identified by a natural number. A \emph{packet} $\pkt$ is a record of
fields containing header values such as source and destination
address, protocol type, and so on. We write $\{
f_1; \dots; f_k \}$ for the type of packets having fields $f_i$ and
use ``dot'' notation to project fields from records. The notation $\{
r~\text{with}~f = v \}$ denotes functional update of $r.f$.

\paragraph*{Forwarding Tables.}
A switch configuration is defined in terms of forwarding rules, where
each rule has a \emph{pattern} $\pat$ specified as a record of
optional packet header fields and a port, a list of \emph{actions}
$\act$ that either forward a packet out a given port ($\fwd~pt$) or
modify a header field ($\set{f}{n}$), and a priority that
disambiguates rules with overlapping patterns.  We write $\{\pt?;
f_1?; \dots; f_k? \}$ for the type of patterns, where the question
mark denotes an option type. A set of such rules $\ruls$ forms a
forwarding table $\tbl$. The semantic function $\denot{\tbl}$ maps
packet-port pairs to multisets of such pairs, finding the
highest-priority rule whose pattern matches the packet and applying
the corresponding actions. If there are multiple matching rules with
the same priority, the function is free to pick any of them, and if
there are no matching rules, it drops the packet. The forwarding
tables collectively define the network's \emph{data plane}.

\paragraph*{Commands.}
The \emph{control plane} modifies the data plane by issuing commands
that update forwarding tables. The command $(\sw,\tbl)$ replaces the
forwarding table on switch $\sw$ with $\tbl$ (we call this a {\it
switch-granularity} update).  We model this command as an atomic
operation (it can be implemented with OpenFlow {\it
bundles}~\cite{ONFxx}). Sometimes switch granularity is too coarse to
find an update sequence, in which case one can update
individual rules ({\it rule-granularity}). Our tool supports this
finer-grained mode of operation, but since it is not conceptually
different from switch granularity,
we frame most of our
discussion in terms of {\it switch-granularity}.

To synchronize updates involving multiple switches, we include a
$\wait$ command. In the model, the controller maintains a
natural-number counter known as the current epoch $\Epoch$. Each
packet is annotated with the epoch on ingress. The control command
$\epoch$ increments the epoch so that subsequent incoming packets are
annotated with the next epoch, and $\flush$ blocks the controller until all packets
annotated with the previous epoch have exited the network. We
introduce a command $\wait$ defined as $\epoch;\flush$.  The epochs
are included in our model solely to enable reasoning. They do not need
to be implemented in a real network---all that is needed is a
mechanism for blocking the controller to allow a flush of all packets currently
in the network. For example, given a topology, one could compute a
conservative delay based on the maximum hop count, and then implement
$\wait$ by sleeping, rather than synchronizing with each switch.
Note that we implicitly assume failure-freedom and packet-forwarding fairness
of switches and links, i.e. there is an upper bound on each element's packet-processing
time.

\paragraph*{Elements.}
The elements $E$ of the network model include switches $S_i$,
links $L_j$, and a single controller element $C$,
and a \emph{network} $N$ is a tuple containing these. 
Each switch
$S_i$ is encoded as a record comprising a unique identifier $\sw$, a
table $\tbl$ of prioritized forwarding rules, and a multiset $\pairs$
of pairs $(\pkt,\pt)$ of buffered packets and the ports they should be
forwarded to respectively.
Each link $L_j$ is represented by a record
consisting of two \emph{locations} $\loc$ and $\loc'$ and a list of
queued packets $\pkts$,
where a location is either a host
or a switch-port pair.
Finally, controller $C$ is represented by a record
containing a list of commands $\cmds$ and an epoch $\Epoch$.
We assume that commands are totally-ordered.
The controller can ensure this by using OpenFlow {\it barrier} messages.

\paragraph*{Operational semantics.}
Network behavior is defined by small-step
operational rules in Figure~\ref{fig:model}.
These define interactions between subsets of elements, based on 
OpenFlow semantics~\cite{mckeown2008openflow}. States of the
model are given by multisets of elements. We write $\onebag{x}$
to denote a singleton multiset, and $m_1 \bagplus m_2$ for the union of
multisets $m_1$ and $m_2$. We write $\onelist{x}$ for a
singleton list, and $l_1 \listplus l_2$ for concatenation of $l_1$
and $l_2$. Each transition $N \steps{o} N'$ is annotated,
with $o$ being either an empty annotation, or an {\em observation} $(\sw,\pt,\pkt)$
indicating the location and packet being processed.

The first rules describe date-plane behavior.
The \textsc{In} rule admits arbitrary packets into the network from a host,
stamping them with the current controller epoch.
The \textsc{Out} rule removes a packet buffered on a link adjacent to
a host. \textsc{Process} processes a
single packet on a switch, finding the highest priority rule with
matching pattern, applying the actions of that rule to
generate a multiset of packets, and adding those packets to the
output buffer. \textsc{Forward}
moves a packet from a switch to the adjacent
link. The final rules describe control-plane behavior.
\textsc{Update} replaces the table on a
single switch. \textsc{Incr} increments the epoch 
on the controller, and \textsc{Flush} blocks the controller
until all packets in the network are annotated with {\it at least}
the current epoch ($\Epoch(\mathit{Es})$
denotes the smallest annotation on any packet in $\mathit{Es}$).
Finally, \textsc{Congruence}, allows any sub-collection of network elements to interact.

\subsection{Network Update Problem}

In order to define the network update problem, we need to first define
{\it traces} of packets flowing through the network.

\paragraph*{Packet traces.}
Given a network $N$, our
operational rules can generate sequences of observations. However,
the network can process many packets concurrently, and we want
observations generated by a single packet.
We define a successor relation $\sqsubseteq$ for observations (Definition~\ref{def:succ},
Appendix~\ref{appnet}).
Intuitively $o \overset{\Epoch}{\sqsubseteq} o'$ if the
network can directly produce the packet in $o'$ by
processing $o$ in the epoch $\Epoch$.

\begin{definition}[Single-Packet Trace]
\label{defsingle}
Let $N$ be a network. A sequence $(o_1 \cdots o_l)$ is
  a \emph{single-packet trace of $N$} if \(
  N \steps{o_1'} \dots \steps{o_k'} N_k \) such that $(o_1 \cdots
  o_l)$ is a subsequence of $(o_1' \cdots o_k')$ for which
\setlength{\pltopsep}{0.1em}
\setlength{\plitemsep}{0.0em}
\begin{compactitem}
\item every
  observation is a successor of the preceding observation in
  monotonically increasing epochs, and
\item if $o_1=o_j'=(\sw,\pt,\pkt)$,
then $\exists o_i' \in \{o_1', \cdots, o_{j-1}'\}$ such that
the $o_i'$ transition is an {\sc In} moving $\pkt$ from host to $(\sw,\pt)$ and
none of $o_{i}', \cdots, o_{j-1}'$ is a predecessor of $o_1$, and
\item the $o_l$ transition is an {\sc Out} terminating at a host.
\end{compactitem}
\end{definition}

\noindent Intuitively, single-packet traces are end-to-end paths through the network.  
We write $\Traces{N}$ for the set of single-packet traces generated by
$N$. A trace $(o_1 \cdots o_k)$ is \emph{loop-free} if $o_i \neq o_j$
for all distinct $i$ and $j$ between $1$ and $k$. We
consider only loop-free traces, since a network that forwards packets around
a loop is generally considered to be misconfigured. In the worst case,
forwarding loops can cause a packet storm, wasting bandwidth and degrading
performance. Our tool automatically detects/rejects such configurations.

\paragraph*{LTL formulas.} 

Many important network properties can be understood by reasoning about
the traces that packets can take through the network. For example,
reachability requires that all packets starting at $src$ eventually
reach $dst$. Temporal logics are an expressive and well-studied
language for specifying such trace-based properties. Hence, we use
Linear Temporal Logic (LTL) to describe traces in our network
model. Let $\AP$ be atomic propositions that test the value of
a switch, port, or packet field: $f_i = n$. We call elements of the set
$2^{\AP}$ {\em traffic classes}. Intuitively, each
traffic class $T$ identifies a set of packets that agree on the values
of particular header fields. An LTL formula $\varphi$ in negation
normal form (NNF) is either $true$, $false$, atomic proposition $p$ in
$\AP$, negated proposition $\neg p$, disjunction
$\varphi_1 \vee \varphi_2$, conjunction $\varphi_2 \wedge \varphi_2$,
next $X \varphi$, until $\varphi_1 U \varphi_2$, or release $\varphi_1
R \varphi_2$, where $\varphi_1$ and $\varphi_2$ are LTL formulas in
NNF. The operators $F$ and $G$ can be defined using other
connectives. Since (finite) single-packet traces can be viewed as
infinite sequences of packet observations where the final observation
repeats indefinitely, the semantics of the LTL formulas can be defined
in a standard way over traces.  We write $t \models \varphi$ to
indicate that the single-packet trace $t$ satisfies the formula
$\varphi$ and $\Trace \models
\varphi$ to indicate that $t \models \varphi$ for each $t$ in
$\Trace$. Given a network $N$ and a formula $\varphi$, we write
$N \models \varphi$ if $\Traces{N} \models \varphi$.

\paragraph*{Problem Statement.} 

Recall that our network
model includes commands for updating a single switch, incrementing the
epoch, and waiting until all packets in the preceding epoch have been
flushed from the network. At a high-level, our goal is to identify a
sequence of commands to transition the network between configurations
without violating specified invariants.
First, we need a bit of notation. Given a network $N$, we write
$N[\sw \gets \tbl]$ for the \emph{switch update} obtained by updating
the forwarding table for switch $\sw$ to $\tbl$.
We call $N$ {\em static} if $C.cmds$ is empty.
If static networks $N_1, N_n$ have the same traces
$\Trace(N_1) = \Trace(N_n)$, then we say they are trace-equivalent,
$N_1 \simeq N_n$.

\begin{definition}[Network Update]
Let $N_1$ be a static network. A command sequence $\cmds$ induces a
sequence $N_1, \dots, N_n$ of static networks if $c_1 \cdots c_{n-1}$
are the update commands in $\cmds$, and
for each $c_i=(\sw,\tbl)$, we have $N_i[\sw \gets \tbl] \simeq N_{i+1}$.
\end{definition}

\noindent We write $\Step{N_1}{\cmds}{N_n}$ if there exists such
a sequence of static networks induced by $\cmds$ which ends with $N_n$.

We call $N$ {\em stable} if all packets in $N$ are annotated with
the same epoch.
Intuitively, a stable network is one with no in-progress update,
i.e. any preceding update command was finalized with a {\it wait}.
Consider the set of {\it unconstrained} single-packet traces
generated by removing the requirement 
that traces start at an ingress (see Definition~\ref{defsingle_all},
Appendix~\ref{appnet}). 
This includes $\Traces{N}$ as well as traces of packets
initially present in $N$. We call this $\TracesAll{N}$, and note
that for a {\em stable} network $N$, $\TracesAll{N}$ is equal to
$\Traces{N}$.

\begin{definition}[Update Correctness]
Let $N$ be a stable static network and let $\varphi$ be an LTL formula. The
command sequence $\cmds$ is {\em correct} with respect to $N$ and $\varphi$ if
$\hat{N} \models \phi$ where $\hat{N}$ is obtained from $N$ by setting
$C.\cmds = \cmds$.
\end{definition}

A {\it network configuration} is a static network which contains no packets.
\noindent We can now present the problem statement. 

\begin{definition}[Update Synthesis Problem] 
Given stable static network $N$, network configuration $N'$, and
LTL specification 
$\varphi$, construct a sequence of commands $\cmds$
such that (i) $\Step{N}{\cmds}{N''}$ where $N'' \simeq N'$, and (ii) $\cmds$ is correct
with respect to $\varphi$.
\end{definition}

\subsection{Efficiently Checking Network Properties}

To facilitate efficient checking of network properties via LTL model
checkers, we show how to model a network as a Kripke structure.

\paragraph*{Kripke structures.} 
A {\em Kripke structure} is a tuple $(Q,Q_0,\delta,\lambda)$, where
$Q$ is a finite set of states, $Q_0 \subseteq Q$ is a set of initial
states, $\delta \subseteq Q \times Q$ is a transition relation, and
$\lambda: Q \to 2^\AP$ labels each state with a set of atomic
propositions drawn from a fixed set $\AP$.
A Kripke structure is {\em complete} if every state has at least one
successor. A state $q \in Q$ is a {\em sink state} if for all states
$q'$, $\delta(q,q')$ implies that $q=q'$, and we call a Kripke
structure {\em DAG-like} if the only cycles are self-loops on
sink states. In this paper, we will consider complete and
DAG-like Kripke structures.
A {\em trace} $t$ is an infinite sequence of states, $t_0 t_1 \ldots
$ such that $\forall i \geq 0: \delta(t_i,t_{i+1})$.  Given a trace
$t$, we write $t^i$ for the suffix of $t$ starting at the $i$-th
position---i.e., $t^i = t_i t_{i+1} \ldots$. Given a set of traces
$\Trace$, we let $\Trace^i$ denote the set $\{t^i \mid t \in \Trace\}$.
Given a state $q$ of a Kripke structure $K$, let $\traces_K(q)$
be the set of traces of $K$ starting from $q$ and $\succs_K(q)$ be the
set of states defined by $q' \in \succs_K(q)$ if and only if
$\delta(q,q')$. 
We will omit the subscript $K$ when it is clear form the context. 
A Kripke structure $K=(Q,Q_0,\delta,\lambda)$ satisfies
an LTL formula $\varphi$ if for all states $q_0 \in Q_0$ we have that
$\traces(q_0) \models \varphi$.

\paragraph{Network Kripke structures.}
For every static $N$, we can generate a Kripke structure
$\Kripke{N}$ containing traces which correspond according
to an intuitive trace relation $\lesssim$
(Definition~\ref{def:netkripke},~\ref{def:traceequiv}, Appendix~\ref{appnet}). 
We currently do not reason about packet modification, 
so the Kripke structure has disjoint parts corresponding
to the traffic classes. It is straightforward to 
enable packet modification, by adding transitions
between the parts of the Kripke structure, 
but we leave this for future work.
We now show that the generated Kripke structure
faithfully encodes the network semantics.
\begin{lemma}[Network Kripke Structure Soundness]
\label{lem:kripke:sound}
Let $N$ be a static network and $K = \Kripke{N}$ a network Kripke
structure. For every single-packet trace $t$ in $\Traces{N}$ there
exists a trace $t'$ of $K$ from a start state such that $t \lesssim
t'$, and vice versa.
\end{lemma}

\noindent 
This means that
checking LTL over single-packet traces can be performed via
LTL model-checking of Kripke structures.

\paragraph*{Checking network configurations.} 
One key challenge arises because the network
is a distributed system. Packets can 
``see'' an inconsistent configuration (some switches updated, some
not), and reasoning about possible interleavings 
of commands becomes intractable in this context. We can simplify
the problem if we ensure that each packet traverses at most one switch
that was updated after the packet entered the network.
\begin{definition}[Careful Command Sequences]
A sequence of commands $(\cmd_1 \cdots \cmd_n)$ is \emph{careful} if
every pair of switch updates is separated by a $\wait$ command.
\end{definition}

\noindent In the rest of this paper, we consider careful
command sequences, and develop a sound and complete algorithm
that finds them efficiently. Section~\ref{synthesis} 
describes a technique for removing wait commands that works well
in practice, but we leave {\em optimal} wait
removal for future work. 
Recall that $\Traces{N}$ denotes the sequence of all traces that a packet could take through the network,
regardless of when the commands in $N.\cmds$ are executed. This is a
superset of the traces induced by each static $N_i$ in a solution
to the network update problem. However, if $\cmds$ is careful,
then each packet only encounters a single configuration, allowing the correctness of the sequence to be
reduced to the correctness of each $N_i$.

\begin{lemma}[Careful Correctness]
\label{lem:careful:corr}
Let $N$ be a stable network with $C.\cmds$ careful and let $\varphi$ be an LTL
formula. If $\cmds$ is careful
and $N_i
\models \phi$ for each static network in any sequence induced by
$\cmds$,
then $\cmds$ is correct with respect to $\varphi$.
\end{lemma} 

\noindent In Lemmas~\ref{lem:trace:eq} and \ref{lem:induce:seq} (Appendix~\ref{appnet}), we
show that checking the {\it unique sequence of network configurations}
induced by $\cmds$ is equivalent to the above. Next we will develop a
sound and complete algorithm that solves the update synthesis problem
for careful sequences by checking configurations.

\begin{figure}[t!]
{\footnotesize
\begin{algorithmic}[1]
\Statex{\hspace*{-5.5mm}\textbf{Procedure}~\textsc{OrderUpdate}($\NetPolicy_i,\NetPolicy_f,\varphi$)} 
\Require~Initial static network $\NetPolicy_i$, final static
configuration $\NetPolicy_f$, formula $\varphi$.    
\Ensure~update sequence $L$, or error $\epsilon$ if no update
sequence exists
  \State $W \gets false$ \Comment{Formula encoding wrong configurations.}   
  \State $V \gets false$ \Comment{Formula encoding visited configurations.} 
  \State $(ok,L) \gets \textsc{DFSforOrder}(\NetPolicy_i,\Kripke{\NetPolicy_i},\bot,\varphi,\labeling_0)$ 
  \If {$ok$} \Return $L$
  \Else~\Return $\epsilon$ \Comment{Failure---no update exists.}
  \EndIf
\end{algorithmic}
\begin{algorithmic}[1]
\makeatletter\setcounter{ALG@line}{5}\makeatother
\Statex{\hspace*{-5.5mm}\textbf{Procedure}~\textsc{DFSforOrder}($\NetPolicy$,$K$,$s$,$\varphi$,$\labeling$)} 
\Require~Static network $\NetPolicy$ and Kripke structure $K$, next
  switch to update $s$, formula $\varphi$, 
and labeling $\labeling$.
\Ensure~Boolean ok if a correct update exists; correct update sequence $L$
\If {$\NetPolicy \models V \vee W$} \Return $(\mathit{false},[\,])$
\EndIf
\If {$s = \bot$} 
$(ok,cex,\labeling) \gets \mathit{modelCheck}(K,\varphi)$ \label{line:modelcheck} 
\Else
\State $(\NetPolicy,K,S) \gets \swUpdate(\NetPolicy,s)$
\State $(ok,cex,\labeling) \gets
\mathit{incrModelCheck}(K,\varphi,S,\labeling)$ \label{line:incrmodelcheck} 
\EndIf
\State $V \gets V \vee \mathit{makeFormula}(\NetPolicy)$
\If{$\neg ok$} 
   \State $W \gets W \vee \mathit{makeFormula}(cex)$ \label{line:ctrex}
   \State \Return $(\mathit{false},[\,])$
\EndIf
\If {$\NetPolicy = \NetPolicy_f$} \Return $(\mathit{true},[s])$
\EndIf
\For {$s' \in \mathit{possibleUpdates}(\NetPolicy)$ }
\State $(ok,L) \gets \textsc{DFSforOrder}(\NetPolicy,K,s',\varphi,\labeling)$
\If {$ok$} \label{line:addcommand} \Return $(\mathit{true},(\mathit{upd}\ s')::\mathit{wait}::L)$
\EndIf
\EndFor 
\State \Return $(\mathit{false},[\,])$
\end{algorithmic}
}
\caption{\textsc{OrderUpdate} Algorithm.}
\label{algo:order}
\spacehack{-.75em}
\end{figure}

\optnewpage
\section{Update Synthesis Algorithm \budget{2}}
\label{synthesis}

This section presents a synthesis algorithm that searches through
the space of possible solutions, using counterexamples to detect wrong
configurations and exploiting several optimizations.

\subsection{Algorithm Description}
\textsc{OrderUpdate} (Figure~\ref{algo:order})
returns a {\it simple} sequence of updates
(one in which each switch appears at most once),
or fails if no such sequence exists.
Note that we could broaden our {\it simple} definition,
e.g. {\it k-simple}, where each switch appears at most $k$ times%
, but we have found the above
restriction to work well in practice.
The core procedure is
\textsc{DFSforOrder}, which manages the search and invokes the model
checker (we use DFS because we
expect common properties/configurations to admit many update sequences).
It attempts to add a switch $s$ to the
current update sequence, yielding a new network configuration. We
maintain two formulas, $V$ and $W$, tracking the set of
configurations that have been visited so far, and the set of
configurations excluded by counterexamples. 

To check whether all packet traces in this configuration satisfy the
LTL property $\varphi$, we use our (incremental) model checking
algorithm (discussed in Section~\ref{checking}). 
First, we call
a full check of the model (line~\ref{line:modelcheck}). 
The model checker labels the Kripke structure nodes with
information about what formulas hold for paths
starting at that state. 
The labeling (stored in \labeling) is then re-used in the subsequent
model checking 
calls for related Kripke structures (line~\ref{line:incrmodelcheck}). 
The parameters passed in the incremental model checking call are:
updated Kripke structure $K$, specification $\varphi$,
set of nodes $S$ in $K$ whose transition function has changed
by the update of the switch $s$, and correct
labeling $\labeling$ of the Kripke structure before the update. 
Note that before the initial model checking, we convert the
network configuration $\NetPolicy$ to a Kripke structure $K$. 
The update of $K$ is performed by a
function $\swUpdate$ that returns a triple
$(\NetPolicy',S,K')$, 
where 
$\NetPolicy'$
is the new static network, $K'$ is the updated Kripke structure
obtained as $\Kripke{\NetPolicy'}$, and $S$ is the set of nodes that
have different outgoing transitions in $K'$.

If the model checker returns {\em true}, then $\NetPolicy$ is safe and
the search proceeds recursively, after adding
$(\mathit{upd}\ s')$ to the current sequence of commands. 
If the model checker returns {\em false}, 
the search backtracks, using the counterexample-learning approach below.

\subsection{Optimizations}

We now present optimizations improving synthesis
({\em pruning with counterexamples},  {\em early search
termination}), and improving efficiency of
synthesized updates ({\em wait removal}).

\paragraph*{A. Counterexamples.}
Counterexample-based pruning learns network configurations that do not
satisfy the specification to avoid making future model checking calls
that are certain to fail. The function $\mathit{makeFormula(cex)}$
(Line~\ref{line:ctrex}) returns a formula representing the set of
switches that occurred in the counterexample trace $cex$, with flags
indicating whether each switch was updated. This allows
equivalent future configurations to be eliminated without
invoking the model checker. Recall the red-green example in
Section~\ref{ex:redgreen} and suppose that we update A1 and then
C2. At the intermediate configuration obtained by updating just A1,
packets will be dropped at C2, and the specification
(H1-H3 connectivity) will not be satisfied. The formula for the unsafe
set of configurations that have A1 updated and C2 not updated will be
added to $W$.  In practice, many counterexamples are small compared to
network size, and this greatly prunes the search space. 

\paragraph*{B. Early search termination.}
The early search termination optimization speeds up 
termination of the search when no (switch-granularity) update sequence is possible. 
Recall how we use counterexamples to prune 
{\it configurations}. With similar reasoning, we can use counterexamples
for pruning possible {\it sequences of updates}. Consider a counterexample
trace which involves three nodes $A,B,C$, with $A$ 
updated, $B$ 
updated, and $C$ not updated. This can be seen as requiring
that $C$ must be updated before $A$, or $C$ must be updated
before $B$. Early search termination involves collecting such
constraints on possible updates, and terminating if these constraints
taken together form a contradiction. In our tool, this is done
efficiently using an (incremental) SAT solver. If the solver determines
that no update sequence is possible, the search terminates. 
For simplicity, early search termination is not shown in Figure \ref{algo:order}.

\paragraph*{C. Wait removal.}
This heuristic eliminates waits that are unnecessary for
correctness. 
Consider an update sequence $L = \cmd_0 \cmd_1 \cdots \cmd_n$, and consider some
switch update $\cmd_k = (upd\ s)$. In the configuration resulting from
executing the sequence $\cmd_0 \cmd_1 \cdots \cmd_{k-1}$,
if the switch $s$ cannot possibly receive a packet which passed through
some switch $s_0$ before an update $\cmd_j{=}(upd\ s_0)$ where $j < k$, then we can update $s$
without waiting. Thus, we can remove some
unnecessary waits if we can 
maintain reachability-between-switches information during the update. Wait removal is
not shown in Figure \ref{algo:order}, but in our tool, it
operates as a post-processing pass once an update sequence is found. In practice, this 
removes a majority of unnecessary waits (see \S~\ref{evaluation}).   

\subsection{Formal Properties}
The following two theorems show that our algorithm is sound for
careful updates, and complete if we limit our search to 
{\it simple} update sequences (see Appendix \ref{app:dfs} for proofs).

\begin{theorem}[Soundness]
\label{thm:sound}
Given initial network $\NetPolicy_i$, final configuration
$\NetPolicy_f$, and LTL formula $\varphi$, if \textsc{OrderUpdate}
returns a command sequence $\cmds$, then
$\Step{\NetPolicy_i}{\cmds}{\NetPolicy'}$ s.t. $\NetPolicy' \simeq \NetPolicy_f$, and $\cmds$ is correct
with respect to $\varphi$ and $\NetPolicy_i$.
\end{theorem}

\begin{theorem}[Completeness]
\label{prop:complete}
Given initial network $\NetPolicy_i$, final configuration
$\NetPolicy_f$, and specification $\varphi$, if there exists a
simple, careful sequence $\cmds$ with
$\Step{\NetPolicy_i}{\cmds}{\NetPolicy'}$ s.t. $\NetPolicy' \simeq \NetPolicy_f$, then
\textsc{OrderUpdate} returns one such sequence.
\end{theorem}

\optnewpage
\section{Incremental Model Checking \budget{4}}
\label{checking}

We now present an incremental algorithm for model checking Kripke
structures. This algorithm is central to our synthesis tool, which
invokes the model checker on many closely related structures.  The
algorithm makes use of the fact that the only cycles in the Kripke
structure are self-loops on sink nodes---something that is true of
structures encoding loop-free network configurations---and re-labels
the states of a previously-labeled Kripke structure
with the (possibly different) formulas that hold after an update.

\subsection{State Labeling}
We begin with an algorithm for labeling states of a
Kripke structure with sets of formulas, following the approach
of~\cite{WVS83} (WVS) and~\cite{VW86}. The WVS algorithm translates an
LTL formula $\varphi$ 
into a local automaton and an eventuality automaton. The local
automaton checks consistency between a state and its predecessor, and
handles labeling of all formulas except $\varphi_1\ U\ \varphi_2$,
which is checked by the eventuality automaton.  The two automata are
composed into a single B\"uchi automaton whose states
correspond to subsets of the set of subformulas of $\varphi$ and their
negations.
Hence, we label each Kripke state by a set $L$ of
sets of formulas such that if a state $q$ is labeled by $L$, then for
each set of formulas $S$ in $L$, there exists a trace $t$ starting
from $q$ satisfying 
all the formulas in $S$. 

We now describe state labeling precisely. Let  
$\varphi$ be an LTL formula in NNF. 
The \emph{extended closure} of 
$\varphi$, written $\ecl(\varphi)$, is the set of all subformulas of
$\varphi$ and their negations:
\setlength{\pltopsep}{0.25em}
\setlength{\plitemsep}{0.0em}
\begin{compactitem}
\item $true \in \ecl(\varphi)$ 
\item $\varphi \in \ecl(\varphi)$
\item If $\psi \in \ecl(\varphi)$, then $\neg \psi \in
  \ecl(\varphi)$ \\ (we identify $\psi$ with $\neg \neg \psi$, for
  all $\psi$). 
\item If $\varphi_1 \vee \varphi_2 \in \ecl(\varphi)$, then $\varphi_1
  \in \ecl(\varphi)$ and  
  $\varphi_2 \in \ecl(\varphi)$. 
\item If $\varphi_1 \wedge \varphi_2 \in \ecl(\varphi)$, then
  $\varphi_1 \in \ecl(\varphi)$ and $\varphi_2 \in \ecl(\varphi)$. 
\item If $X\, \varphi_1 \in \ecl(\varphi)$, then $\varphi_1 \in
  \ecl(\varphi)$. 
\item If $\varphi_1\, U\, \varphi_2 \in \ecl(\varphi)$, then $\varphi_1
  \in \ecl(\varphi)$ and $\varphi_2 \in \ecl(\varphi)$ 
\item If $\varphi_1\, R\, \varphi_2 \in \ecl(\varphi)$, then $\varphi_1
  \in \ecl(\varphi)$ and $\varphi_2 \in \ecl(\varphi)$. 
\end{compactitem}
A subset $M \subset \ecl(\varphi)$ of 
the extended closure is said to be \emph{maximally consistent} if it
contains $true$ and is simultaneously closed and consistent under
boolean operations:
\begin{compactitem}
\item $true \in M$
\item $\psi \in M$ iff $\neg \psi \not\in M$
  (we identify $\psi$ with $\neg \neg \psi$, for all $\psi$) 
\item $\varphi_1 \vee \varphi_2 \in M$ iff ($\varphi_1 \in M$ or 
  $\varphi_2 \in M$) 
\item $\varphi_1 \wedge \varphi_2 \in M$ iff 
  ($\varphi_1 \in M$ and $\varphi_2 \in M$)
\end{compactitem}
Likewise, the relation 
$\follows(M_1, M_2)$ captures the notion of successor induced by LTL's
temporal operators, lifted to maximally-consistent sets.
We say $\follows(M_1, M_2)$ holds if and only if
all of the following hold:
\begin{compactitem}
\item $X\ \varphi_1 \in M_1$ iff $\varphi_1 \in M_2$ 
\item $\varphi_1\, U\, \varphi_2 \in M_1$ iff \big($\varphi_2
  \in M_1$ $\lor$ ($\varphi_1 \in M_1$ $\land$ $\varphi_1\, U\, \varphi_2 \in
  M_2$)\big)
\item $\varphi_1\, R\, \varphi_2 \in M_1$ iff \big($\varphi_1
  \in M_1$ $\lor$ ($\varphi_2 \in M_1$ $\land$ $\varphi_1\, R\, \varphi_2 \in
  M_2$)\big)
\end{compactitem}
Given a trace $t$ and
a maximally-consistent set $M$, we write $t \models M$ if and only if
for all $\psi \in M$, we have $t \models \psi$.

For the rest of this section, we fix a Kripke structure $K =
(Q,Q_0,\delta,\lambda)$, a state $q$ in $Q$, an LTL formula $\varphi$
in NNF, and a maximally-consistent set $M \subset \ecl(\varphi)$.

To compute the label of a state $q$, there are two cases depending on
whether it is a sink state or a non-sink state. If $q$ is a sink
state, the function $\HoldsSink(q,M)$ computes a predicate that is
true if and only if, for all $\psi \in M$ and the unique trace $t$
starting from $q$, we have $t \models \psi$. More formally,
$\HoldsSink(q, M)$ is defined to be $(\forall \psi \in M:
\HoldsS(q,\psi))$, where $\HoldsS$ is defined as
in Figure \ref{holds_0}.
\begin{figure}[tb]
\begin{displaymath}
\begin{array}{rcl}
 \HoldsS(q,p) & = & q \models p \\
 \HoldsS(q,\lnot p) & = & q \not\models p \\
 \HoldsS(q,\phi_1 \land \phi_2) & = & \HoldsS(q,\phi_1) \land \HoldsS(q,\phi_2) \\
 \HoldsS(q,\phi_1 \lor \phi_2) & = & \HoldsS(q, \phi_1) \lor \HoldsS(q, \phi_2) \\
 \HoldsS(q,\textrm{X}\phi) & = & \HoldsS(q, \phi) \\
 \HoldsS(q,\phi_1 \textrm{ U } \phi_2) & = & \HoldsS(q, \phi_2) \\
 \HoldsS(q,\phi_1 \textrm{ R } \phi_2) & = & \HoldsS(q,\phi_1) \lor \HoldsS(q,\phi_2)
\end{array}
\end{displaymath}
\vspace{-1.0em}
\caption{The $\HoldsS$ function}
\label{holds_0}
\end{figure}
The function $\HoldsS$ computes a predicate that
is true if and only if $\psi$ 
holds at $q$. For example, $\HoldsS(q,\phi_1 \textrm{ U } \phi_2)$ is
defined as $\HoldsS(q, \phi_2)$ because the only transition from $q$
is a self-loop.

For the second case, suppose $q$ is a non-sink state. If we are given
a labeling for $\succs_K(q)$ (the successors of the node $q$), we can
extend it to a labeling 
for $q$. Let $V \subseteq Q$ be a set of vertices. A function
$\labelGraph_K$ is a \emph{correct labeling of $K$ with respect to
$\varphi$ and $V$} if for every $v \in V$, it returns a set $L$ of
maximally consistent sets such that (a) $M \in L$ if and only if
$M \subset \ecl(\varphi)$, and (b) there exists a trace $t$ in
$\traces(v)$ such that $t \models M$. Suppose that $\labelGraph_K$ is a correct
labeling of $K$ with respect to $\varphi$ and $\succs_K(q)$. The
function $\Holds_K(q,M,\labelGraph_K)$ computes a predicate that is
true if and only if there exists a trace $t$ in $\traces_K(q)$ with
$t \models M$. Formally, $\Holds_K(q, M,\labelGraph_K)$  
is defined as $(\lambda(q) = (\AP \cap M))~\wedge~\exists
q' \in \succs_K(q), M' \in \labelGraph_K(q'): \follows(M,M')$. 

The following captures the correctness of labeling:
\begin{lemma}\label{prop:label1state}
First, $\HoldsSink(q,M) \Leftrightarrow \exists t \in \traces(q):
t \models M$ for sink states $q$. Second, if $\labelGraph_K$ is a
correct labeling with respect to $\varphi$ and $\succs_K(q)$,
then \( \Holds_K(q,M,\labelGraph_K) \Longleftrightarrow \exists
t \in \traces_K(q): t \models M \). 
\end{lemma}

Finally, we define $\labelNode_K(\varphi,q,\labelGraph_K)$, which
computes a label $L$ for $q$ such that $M \in L$ if and only if there
exists a trace $t \in \traces_K(q)$ such that $t \models M$ for all $M
\subset \ecl(\varphi)$. We assume that $\labelGraph_K$ is a correct
labeling of $K$ with respect to $\varphi$ and $\succs(q)$. For sink
states, $\labelNode_K(\varphi,q,\labelGraph_K)$ returns $\{ M \mid M
\in \ecl(\varphi) \wedge \HoldsSink(q,M) \}$, while for non-sink
states it returns $\{M \mid M \in \ecl(\varphi) \wedge
\Holds_K(q,M,\labelGraph_K)\}$.

\subsection{Incremental algorithm} 
To incrementally model check a modified Kripke structure, we must
re-label its states with the formulas that hold after the update. 

Consider two Kripke structures 
$K=(Q,Q_0,\delta,\lambda)$ and $K'=(Q',Q'_0,\delta',\lambda')$, such 
that $Q_0=Q'_0$. Furthermore, assume that $Q=Q'$, and there is a set
$U \subseteq Q$ such that $\delta$ and $\delta'$ differ only on nodes
in $U$. We call such a triple $(K,K',U)$ an {\em update} of $K$.  

An update $(K,K',U)$ might add or remove edges connected to a (small)
set of nodes, 
corresponding to a change in the rules on a switch. Suppose that
$\labelGraph_K$ is a correct labeling of $K$ with respect to $\varphi$
and $Q$. The incremental model checking problem is defined as follows:
we are given an update $(K,K',U)$, and $\labelGraph_K$, and we want to
know whether $K'$ satisfies $\varphi$. The
\naive~approach is to model check $K'$ without using the labeling
$\labelGraph_K$. We call this the {\em monolithic} approach.  In
contrast, the {\em incremental} approach uses $\labelGraph_K$ (and
thus intuitively re-uses the results of model checking $K$ to
efficiently verify $K'$). 

\paragraph{Example.}
Consider the left side of
Figure~\ref{fig:incremental}, with $H$ the only initial state. Suppose
that the update modifies $J$, and the $\delta'$ relation applied to $J$ only contains
the pair $(J,N)$, and
consider labeling the structure with
formulas $F\ a$, $F\ b$, and $F\ a \vee F\ b$. To simplify the
example, we
label a node by all those formulas which hold for at least one path
starting from the node (note that in the algorithm, a node is labeled by
a set of sets of formulas, rather than a set of formulas). 
We will have that all the nodes are labeled by $F\ a \vee F\ b$, and
in addition the 
nodes $K, I, H, M, J$ contain label $F\ a$, and the 
nodes $L, I, H, N$ contain $F\ b$. Now we want to relabel
the structure after the update (right-hand side).  Given
that the update changes only node $J$, the labeling can only
change for $J$ and its ancestors.  We therefore start labeling node
$J$, and find that it will now be labeled
with $F\ b$ instead of $F\ a$.  Labeling
proceeds to $H$, whose label does not change
(still labeled by all of $F\ a$, $F\ b$, $F\ a \vee F\ b$).  The
labeling process could then stop, even if $H$ has ancestors.

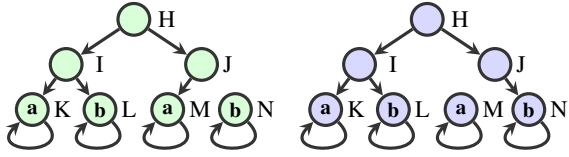
\begin{figure}[t]
\centering
\tikzwrapper{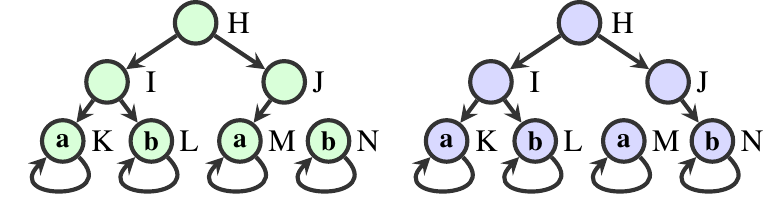}{}{
\begin{tikzpicture}[scale=.60]
\coordinate (C-H1) at (2,0);
\coordinate (C-I1) at (0.5,-1);
\coordinate (C-J1) at (3.5,-1);
\coordinate (C-K1) at (-0.25,-2);
\coordinate (C-L1) at (1.25,-2);
\coordinate (C-M1) at (2.75,-2);
\coordinate (C-N1) at (4.25,-2);
\node[kripke,myempty,label={[label distance=-1]0:H}] (S-H1) at (C-H1) {{\textbf{\footnotesize }}};
\node[kripke,myempty,label={[label distance=1.0]0:I}] (S-I1) at (C-I1) {{\textbf{\footnotesize }}};
\node[kripke,myempty,label={[label distance=-2]0:J}] (S-J1) at (C-J1) {{\textbf{\footnotesize }}};
\node[kripke,myempty,label={[label distance=-2]0:K}] (S-K1) at (C-K1) {{\textbf{\footnotesize a}}};
\node[kripke,myempty,label={[label distance=-2]0:L}] (S-L1) at (C-L1) {{\textbf{\footnotesize b}}};
\node[kripke,myempty,label={[label distance=-2]0:M}] (S-M1) at (C-M1) {{\textbf{\footnotesize a}}};
\node[kripke,myempty,label={[label distance=-2]0:N}] (S-N1) at (C-N1) {{\textbf{\footnotesize b}}};
\draw[black!80] (S-H1) edge[->] (S-I1);
\draw[black!80] (S-H1) edge[->] (S-J1);
\draw[black!80] (S-I1) edge[->] (S-K1);
\draw[black!80] (S-I1) edge[->] (S-L1);
\draw[black!80] (S-J1) edge[->] (S-M1);
\draw[black!80] (S-K1) edge[->,loop below] (S-K1);
\draw[black!80] (S-L1) edge[->,loop below] (S-L1);
\draw[black!80] (S-M1) edge[->,loop below] (S-M1);
\draw[black!80] (S-N1) edge[->,loop below] (S-N1);

\coordinate (C-H2) at (8.5,0);
\coordinate (C-I2) at (7.0,-1);
\coordinate (C-J2) at (10.0,-1);
\coordinate (C-K2) at (6.25,-2);
\coordinate (C-L2) at (7.75,-2);
\coordinate (C-M2) at (9.25,-2);
\coordinate (C-N2) at (10.75,-2);
\node[kripke,myempty2,label={[label distance=-1]0:H}] (S-H2) at (C-H2) {{\textbf{\footnotesize }}};
\node[kripke,myempty2,label={[label distance=1.0]0:I}] (S-I2) at (C-I2) {{\textbf{\footnotesize }}};
\node[kripke,myempty2,label={[label distance=-2]0:J}] (S-J2) at (C-J2) {{\textbf{\footnotesize }}};
\node[kripke,myempty2,label={[label distance=-2]0:K}] (S-K2) at (C-K2) {{\textbf{\footnotesize a}}};
\node[kripke,myempty2,label={[label distance=-2]0:L}] (S-L2) at (C-L2) {{\textbf{\footnotesize b}}};
\node[kripke,myempty2,label={[label distance=-2]0:M}] (S-M2) at (C-M2) {{\textbf{\footnotesize a}}};
\node[kripke,myempty2,label={[label distance=-2]0:N}] (S-N2) at (C-N2) {{\textbf{\footnotesize b}}};
\draw[black!80] (S-H2) edge[->] (S-I2);
\draw[black!80] (S-H2) edge[->] (S-J2);
\draw[black!80] (S-I2) edge[->] (S-K2);
\draw[black!80] (S-I2) edge[->] (S-L2);
\draw[black!80] (S-J2) edge[->] (S-N2);
\draw[black!80] (S-K2) edge[->,loop below] (S-K2);
\draw[black!80] (S-L2) edge[->,loop below] (S-L2);
\draw[black!80] (S-M2) edge[->,loop below] (S-M2);
\draw[black!80] (S-N2) edge[->,loop below] (S-N2);

\end{tikzpicture}}
\spacehack{-.5em}
\caption{Incremental labeling---Initial (left), Final (right)}
\spacehack{-.5em}
\label{fig:incremental}
\end{figure}

\paragraph*{Re-labeling states.}
Let $\ancestors_K(V)$ be the ancestors of $V$ in $K$---i.e., a
set of vertices s.t. $\ancestors_K(V) \subseteq Q$ and
$q \in \ancestors_K(V)$, if some node $v \in V$ is reachable from
$q$. To define incremental model checking for $\varphi$, we 
need a function accepting a property $\varphi$, set
of vertices $V$, labeling $\labelGraph_K$ that is correct for $K$
with respect to $\varphi$ and 
$Q \setminus \ancestors_K(V)$, and returns a correct labeling of $K$
with respect to $\varphi$ and $Q$. This function is:
\[
\begin{array}{l}
\hspace*{-0.1cm}\relabel_K(\varphi,\labelGraph_K,V) =
\begin{cases}
\labelGraph_K & \text{if}~V = \emptyset \\
\relabel_K(\varphi,\labelGraph'_K,V') & \text{otherwise}
\end{cases}
\end{array}
\]
where $\labelGraph'_K(v)$ is $\labelNode_K(\varphi,v,\labelGraph_K)$
if  $v \in V$, and it is $\labelGraph_K(v)$ if $v \not\in V$. The set
$V'$ is $\{q \mid \exists v \in V: v \in \succs_K(q)\}$.

\begin{theorem}\label{thm:relabel}
Let $V \subseteq Q$ be a set of vertices and $\labelGraph_K$ a correct
labeling with respect to $\varphi$ and $Q \setminus \ancestors_K(V)$.
Then $\relabel_K(\varphi,\labelGraph_K,V)$ is a correct labeling w.r.t.
$\varphi$ and $Q$.
\end{theorem}

Given a labeling that is correct with respect to $\varphi$ and $Q$, it
is easy to check whether $\varphi$ is true for all the traces starting
in the initial states: the predicate
$\checkInitStates_K(\labelGraph_K,\varphi)$ is defined as $\forall
q_0 \in Q_0, M \in \labelGraph_K(q_0): \varphi \in M$. Next, let $Q_f$
be the set of all sink states of $K$. Then $\ancestors_K(Q_f)$ is the
set $Q$ of all states $K$. Therefore, for any initial labeling $\labelGraph_K^0$,
$\relabel(\varphi,\labelGraph_K^0,Q_f)$ is a correct labeling with respect to
$\varphi$ and $Q$. The function $\modelCheck_K(\varphi)$ is defined to
be equal to 
$\checkInitStates_K(\relabel_K(\varphi,\labelGraph^0_K,Q_f),\varphi)$, where we can
set $\labelGraph^0_K$ to be the empty labeling $\lambda v. \emptyset$.

We now define our incremental model checking function. Let 
$(K,K',U)$ be an update, and $\labelGraph_K$ a previously-computed
correct labeling of $K$ with respect to $\varphi$ and $Q$, where $Q$
is the set of states of $K$. The function
$\incrModelCheck(K,\varphi,U,\labelGraph_K)$ is defined
as $\checkInitStates_{K'}(\relabel_{K'}(\varphi,\labelGraph_K,U),\varphi)$. 
The following shows the correctness of our model checking functions
(proof of this and the previous theorem are in Appendix \ref{app:incrmodcheck}). 
\begin{corollary}\label{cor:modcheck}
First, $\modelCheck_K(\varphi)=true \iff K \models \varphi$. Second,
for $(K,K',U)$ and $\labelGraph_K$ as above, we have 
$\incrModelCheck(K,\varphi,U,\labelGraph_K)=\mathit{true} \iff
K \models \varphi$.   
\end{corollary}

The runtime complexity of the $\modelCheck_K$ function is 
$O(|K| \times 2^{|\varphi|})$.  The runtime complexity of the $\incrModelCheck$
function is
$O(|\ancestors_K(U)| \times 2^{|\varphi|})$, where $U$ is the set of
nodes being updated.

\paragraph*{Counterexamples.}
This incremental algorithm can generate
counterexamples in cases where the formula does not hold.  A formula
$\neg \varphi$ does not hold if an initial state is labeled by $L$,
such that there exists a set $M \in L$, such that $\neg\varphi \in
M$. Examining the definition of $\labelNode_K$,
we find that in order
to add a set $M$ to the label $L$ of a node $q$, there is a set $M'$
in the label of one a child $q'$ of $q$ that explains why $M$ is in
$L$. The first node of the counterexample trace starting from $q$ is
one such child $q'$.

\optnewpage
\section{Implementation and Experiments \budget{3}}
\label{evaluation}

We have built a prototype tool that implements the algorithms
described in this paper. It consists of 7K lines of OCaml code. 
The system works by building a Kripke structure (\S \ref{model}) and
then repeatedly interacting with a model checker to synthesize an
update. We currently provide four checker backends: {\it Incremental} uses
incremental relabeling to check and recheck formulas, {\it Batch}
re-labels the entire graph on each call, {\it NuSMV} queries a
state-of-the-art symbolic model checker in batch mode, and {\it
NetPlumber} queries an incremental network model
checker~\cite{Kazemian}. All tools except NetPlumber provide
counterexample traces, so our system learns from counterexamples whenever
possible (\S \ref{synthesis}).

\paragraph*{Experiments.}
To evaluate performance, we generated configurations for a variety of
real-world topologies and ran experiments in which we measured the
amount of time needed to synthesize an update (or discover that no
order update exists). These experiments were designed to answer two
key questions: (1) how the performance of our Incremental checker
compares to state-of-the-art tools (NuSMV and NetPlumber), and (2)
whether our synthesizer scales to large topologies. We used the {\it
Topology Zoo} \cite{topozoo} dataset, which consists of $261$ actual
wide-area topologies, as well as synthetically constructed {\it
Small-World}~\cite{newman2001random} and {\it
FatTree}~\cite{alfares08} topologies. We ran the experiments on a
64-bit Ubuntu machine with 20GB RAM and a quad-core Intel i5-4570 CPU
(3.2 GHz) and imposed a 10-minute timeout for each run. We ignored
runs in which the solver died due to an out-of-memory error or
timeout---these are infrequent (less than 8\% of the 996 runs for
Figure~\ref{fig_nusmv}), and our Incremental solver only died in
instances where other solvers did too.

\paragraph*{Configurations and properties.}
A recent paper \cite{liu2013zupdate}
surveyed data-center operators to discover 
common update scenarios, which mostly involve taking
switches on/off-line and migrating traffic between 
switches/hosts. We designed experiments around a similar
scenario. To create configurations, we connected random
pairs of nodes $(s,d)$ via disjoint initial/final paths $W_i,W_f$,
forming a ``diamond'', and asserted one of the following properties for each pair:

\begin{compactitem}
\item {\it Reachability:} traffic from a given source must reach a
  certain
  destination: $(\mathrm{port}=s) \Rightarrow \mathrm{F\,}(\mathrm{port}=d)$ 
\item {\it Waypointing:} traffic must traverse a waypoint $w$: \\
$(\mathrm{port}{=}s) \Rightarrow \big((\mathrm{port}{\not=}d) \mathrm{\,U\,} ((\mathrm{port}{=}w) \land \mathrm{F\,}(port{=}d))\big)$
\item {\it Service chaining:} traffic must waypoint through
  several intermediate nodes: $(\mathrm{port}=s) \Rightarrow way(W,d)$, where \\
\vspace{-1em}
\begin{displaymath}
\begin{array}{rcl}
 way([\,],d) & \equiv & \mathrm{F\,}(\mathrm{port}=d)\\
 way(w_i::W,d) & \equiv & \big((\bigwedge_{w_k \in W}\mathrm{port}{\not=}w_k \,\land\, \mathrm{port}{\not=}d) \\
               &        & \ \mathrm{\,U\,} ((\mathrm{port}=w_i) \land way(W,d))\big).
\end{array}
\end{displaymath}
\vspace{-1em}
\end{compactitem}

\paragraph*{Incremental vs. NuSMV/Batch.} 
Figure~\ref{fig_nusmv}~(a-c) compares the performance of Incremental
and NuSMV backends for the reachability property. Of the 247 Topology
Zoo inputs that completed successfully, our tool solved all of them
faster. The measured speedups were large, with a geometric mean of
$447.23$x. For the 24 FatTree examples, the mean speedup was $465.03$x,
and for the 25 Small-World examples, the mean speedup was $4484.73$x.
We also compared the Incremental and Batch solvers on the same inputs.
Incremental performs better on almost all examples, with mean speedup of
$4.26$x, $5.27$x, $11.74$x on the datasets shown in Figure
\ref{fig_monolithic}(a-c) and maximum runtimes of
$0.36$s, $2.80$s, and $0.92$s respectively. The maximum runtimes for
Batch were $6.71$s, $39.75$s, and $12.50$s.

\newcommand{\myfigure}[9]{
  \begin{axis}[%
    name=#2,
    xlabel style={align=center,yshift=1.0ex%
    }, 
    ylabel style={align=center}, 
    scaled x ticks=false,
    title style={yshift=-1.5ex},
    legend pos=north west,
    width=2.5in, 
    height=1.2in,
    legend style={%
      at={(0.85,0.97)},anchor=north east,
      fill=none,
      font=\tiny,line width=0.5pt,row sep=-1ex},
    xmax=#5, ymax=#6,
    legend cell align=left,
    restrict y to domain=-1000:1000,
    #9 
        ]
\ifx&#3&%
    \def\myfile{../../experiments/data/#1.csv}
    \addplot [no markers, \plotcolor,opacity=0.5,line width=1pt, line cap=round, dash pattern=on 0pt off 2\pgflinewidth,forget plot] gnuplot [raw gnuplot] { 
            f(x) = a*(x**b)+c; 
            a=1; b=#8; c=1; 
            fit f(x) "<awk -F ',' '$14!=0.000 { print $0 }' \myfile" u 10:15 via a,b; 
            plot [x=0:\maxtrend] f(x); 
    };
    \addplot [mark=o,only marks, \plotcolor,mark options={fill=white},mark size=1pt] table[x index=9,y index=14,header=true,col sep=comma,%
    restrict expr to domain={\thisrow{suc_two}}{1:2},%
    ] {\myfile};
    \def\myfile{../../experiments/data/#4.csv}
    \addplot [no markers, blue,opacity=0.5,line width=1pt, line cap=round, dash pattern=on 0pt off 2\pgflinewidth,forget plot] gnuplot [raw gnuplot] { 
            f(x) = a*(x**2)+c; 
            a=1; b=#8; c=0.01; 
            fit f(x) "<awk -F ',' '$14!=0.000 { print $0 }' \myfile" u 10:15 via a,c; 
            plot [x=0:\maxtrendx] f(x); 
    };
    \addplot [mark=square,blue,only marks,mark options={fill=white},mark size=1pt] table[x index=9,y index=14,header=true,col sep=comma,%
    restrict expr to domain={\thisrow{suc_two}}{1:2},%
    ] {\myfile};
    \addplot [no markers, violet,opacity=0.5,line width=1pt, line cap=round, dash pattern=on 0pt off 2\pgflinewidth,forget plot] gnuplot [raw gnuplot] { 
            f(x) = a*(x**2)+c; 
            a=1; b=#8; c=0.01; 
            fit f(x) "<awk -F ',' '$2!=0.000 { print $0 }' \myfile" u 10:4 via a,c; 
            plot [x=0:\maxtrendx] f(x); 
    };
    \addplot [mark=+,only marks, violet,mark options={fill=white},opacity=0.75] table[x index=9,y index=3,header=true,col sep=comma,%
    restrict expr to domain={\thisrow{suc_one}}{1:2},%
    ] {\myfile};
\else
    \def\myfile{../../experiments/data/#3.csv}
    \addplot [no markers, orange,opacity=0.5,line width=1pt, line cap=round, dash pattern=on 0pt off 2\pgflinewidth,forget plot] gnuplot [raw gnuplot] { 
            g(x) = d*(x**2)+e; 
            d=1; b=#8; e=0.1; 
            fit g(x) "<awk -F ',' '$14!=0.000 { print $0 }' \myfile" u 10:15 via d,e; 
            plot [x=0:\maxtrend] g(x); 
    };
    \addplot [mark=triangle,orange,only marks,mark options={fill=white},mark size=1.5pt] table[x index=9,y index=14,header=true,col sep=comma,%
    restrict expr to domain={\thisrow{suc_two}}{1:2},%
    ] {\myfile};
    \addplot [no markers, darkgreen,opacity=0.5,line width=1pt, line cap=round, dash pattern=on 0pt off 2\pgflinewidth,forget plot] gnuplot [raw gnuplot] { 
            f(x) = a*(x**2)+c; 
            a=1; b=#8; c=0.1; 
            fit f(x) "<awk -F ',' '$2!=0.000 { print $0 }' \myfile" u 10:4 via a,c; 
            plot [x=0:\maxtrendx] f(x); 
    };
    \addplot [mark=x,only marks, darkgreen,mark options={fill=white},opacity=0.75] table[x index=9,y index=3,header=true,col sep=comma,%
    restrict expr to domain={\thisrow{suc_one}}{1:2},%
    ] {\myfile};
\fi
\node [below left] at (axis cs:  \pgfkeysvalueof{/pgfplots/xmax},  \pgfkeysvalueof{/pgfplots/ymax}) {{\bf(#2)}};

    \typeout{-----------------}
    \typeout{PROCESSING: \myfile}
    \legend{#7}
  \end{axis}
}

\newcommand{\hpgfspace}{3em}
\newcommand{\vpgfspace}{1.5em}
\begin{figure*}
\centering
\tikzwrapper{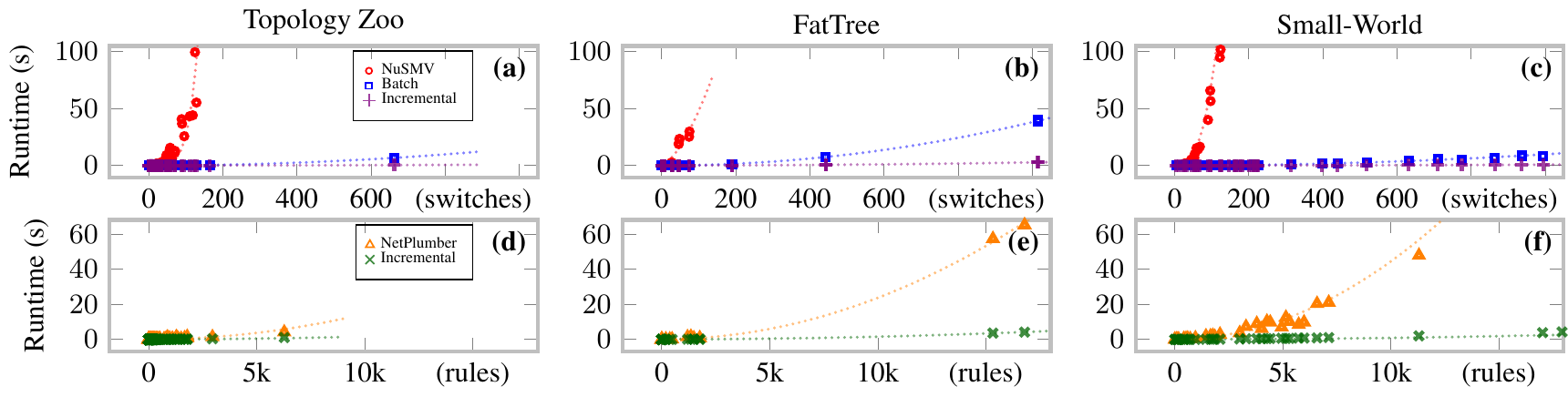}{}{
\begin{tikzpicture}
\def\plotcolor{red}
\def\maxtrend{1050} \def\maxtrendx{900}
\myfigure{result_reach_topozoo_nusmv_all}{a}{}{result_way_topozoo_batch_big}{1050}{105}{NuSMV \\ Batch \\ Incremental \\}{2}{title={Topology Zoo},xtick={0,200,400,600,800,1000},xticklabels={0,200,400,600,\hspace{2em}(switches)},xlabel={},ylabel=Runtime (s),at={($(\hpgfspace,0)$)},anchor=west}
\def\maxtrend{140} \def\maxtrendx{1050}
\myfigure{result_reach_fattree_nusmv}{b}{}{result_way_fattree_batch}{1050}{105}{}{2}{title={FatTree},xlabel={},xtick={0,200,400,600,800,1000},xticklabels={0,200,400,600,\hspace{2em}(switches)},at={($(a.east)+(\hpgfspace,0)$)},anchor=west}
\def\maxtrend{1050} \def\maxtrendx{1050}
\myfigure{result_reach_smallworld_nusmv}{c}{}{result_way_smallworld_batch}{1050}{105}{}{2}{title={Small-World},xlabel={},xtick={0,200,400,600,800,1000},xticklabels={0,200,400,600,\hspace{2em}(switches)},at={($(b.east)+(\hpgfspace,0)$)},anchor=west}

\def\plotcolor{orange}
\def\maxtrend{9000} \def\maxtrendx{9000}
\myfigure{}{d}{result_reach_topozoo_netplumber_all}{}{18000}{68}{NetPlumber \\ Incremental \\}{3}{xlabel={},ylabel=Runtime (s),xticklabels={0,0,5k,10k,(rules)},at={($(a.south)-(0,\vpgfspace)$)},anchor=north}
\def\maxtrend{18000} \def\maxtrendx{18000}
\myfigure{}{e}{result_reach_fattree_netplumber}{}{18000}{68}{}{3}{xlabel={},xticklabels={0,0,5k,10k,(rules)},at={($(d.east)+(\hpgfspace,0)$)},anchor=west}
\def\maxtrend{18000} \def\maxtrendx{18000}
\myfigure{}{f}{result_reach_smallworld_netplumber}{}{18000}{68}{}{3}{xlabel={},xticklabels={0,0,5k,10k,(rules)},at={($(e.east)+(\hpgfspace,0)$)},anchor=west}

\end{tikzpicture}}
\vspace{-0.5em}
\caption{Relative performance results: 
{\bf(a-c)} Performance of Incremental vs. NuSMV, Batch, NetPlumber solvers on Topology Zoo, FatTree, Small-World topologies (columns);
{\bf(d-f)} Performance of Incremental vs. NetPlumber (rule-granularity).
}
\label{fig_nusmv}
\label{fig_netplumber}
\label{fig_monolithic}
\end{figure*}

\renewcommand{\vpgfspace}{1.5em}
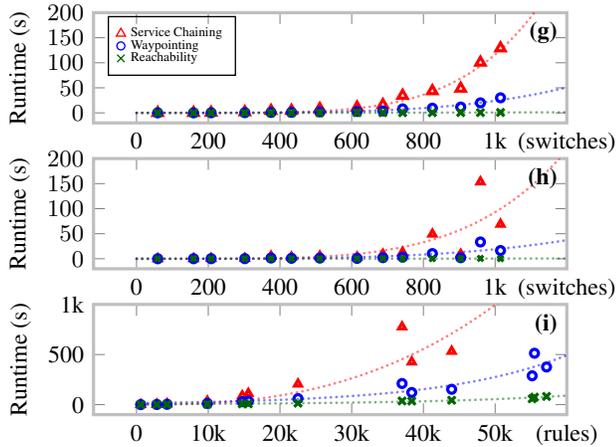
\begin{figure}
\centering
\tikzwrapper{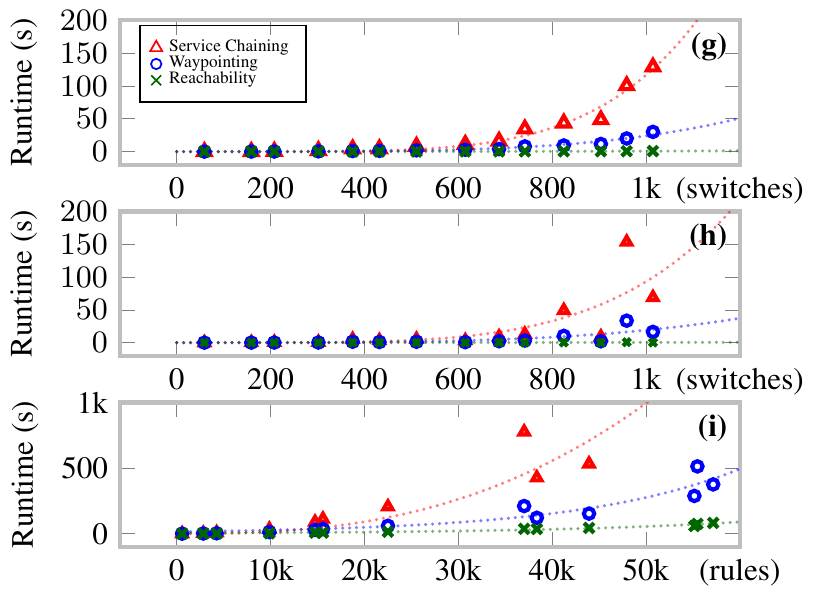}{}{
\begin{tikzpicture}
  \begin{axis}[%
    name=result_way3_smallworld_incr,
    xlabel style={align=center,yshift=1.0ex}, %
    ylabel={Runtime (s)},
    ylabel style={align=center},
    title style={yshift=-1.5ex},
    legend pos = north west,%
    scaled x ticks=false,
    xticklabels={X,0,200,400,600,800,1k,(switches)},
    width=3.1in, height=1.2in,
    ymax=200,xmax=1200,
    legend style={fill=none,font=\tiny,line width=0.5pt,row sep=-1ex},
    legend cell align=left,
        ]
    \def\myfile{../../experiments/data/result_way3_smallworld_incr.csv}
    \addplot [mark=triangle,only marks, red,mark options={fill=white},mark size=2pt] table[x index=9,y index=14,header=true,col sep=comma,%
    restrict expr to domain={\thisrow{suc_two}}{1:2},%
    restrict expr to domain={\thisrow{suc_one}}{1:2},%
    ] {\myfile};
    \typeout{-----------------}
    \typeout{PROCESSING: \myfile}
    \addplot [no markers, red,opacity=0.5,line width=1pt, line cap=round, dash pattern=on 0pt off 2\pgflinewidth,forget plot] gnuplot [raw gnuplot] { 
            f(x) = a*(x**b); 
            a=1; b=3; 
            fit f(x) '\myfile' u 10:15 via a,b; 
            plot [x=0:1200] f(x); 
    };
    \def\myfile{../../experiments/data/result_way_smallworld_incr.csv}
    \addplot [mark=o,only marks, blue,mark options={fill=white},mark size=1.5pt] table[x index=9,y index=14,header=true,col sep=comma,%
    restrict expr to domain={\thisrow{suc_two}}{1:2},%
    restrict expr to domain={\thisrow{suc_one}}{1:2},%
    ] {\myfile};
    \typeout{-----------------}
    \typeout{PROCESSING: \myfile}
    \addplot [no markers, blue,opacity=0.5,line width=1pt, line cap=round, dash pattern=on 0pt off 2\pgflinewidth,forget plot] gnuplot [raw gnuplot] { 
            f(x) = a*(x**b); 
            a=1; b=3; 
            fit f(x) '\myfile' u 10:15 via a,b; 
            plot [x=0:1200] f(x); 
    };
    \def\myfile{../../experiments/data/result_reach_smallworld_incr.csv}
    \addplot [mark=x,only marks, darkgreen,mark options={fill=white}] table[x index=9,y index=14,header=true,col sep=comma,%
    restrict expr to domain={\thisrow{suc_two}}{1:2},%
    restrict expr to domain={\thisrow{suc_one}}{1:2},%
    ] {\myfile};
    \typeout{-----------------}
    \typeout{PROCESSING: \myfile}
    \addplot [no markers, darkgreen,opacity=0.5,line width=1pt, line cap=round, dash pattern=on 0pt off 2\pgflinewidth,forget plot] gnuplot [raw gnuplot] { 
            f(x) = a*(x**b); 
            a=1; b=3; 
            fit f(x) '\myfile' u 10:15 via a,b; 
            plot [x=0:1200] f(x); 
    };
\node [below left] at (axis cs:  \pgfkeysvalueof{/pgfplots/xmax},  \pgfkeysvalueof{/pgfplots/ymax}) {{\bf(g)}};
    \legend{Service Chaining \\ Waypointing \\ Reachability \\}
  \end{axis}
  \begin{axis}[%
    name=result_way3_smallworld_incr_imp,
    at={($(result_way3_smallworld_incr.south)-(0,\vpgfspace)$)},anchor=north,
    xlabel style={align=center,yshift=1.0ex}, %
    ylabel={Runtime (s)},
    ylabel style={align=center},
    title style={yshift=-1.5ex},
    legend pos = north west,%
    scaled x ticks=false,
    xticklabels={X,0,200,400,600,800,1k,(switches)},
    width=3.1in, height=1.2in,
    ymax=200,xmax=1200,
    legend style={fill=none, font=\tiny},
    legend cell align=left,
        ]
    \def\myfile{../../experiments/data/result_way3_smallworld_incr_imp.csv}
    \addplot [mark=triangle,only marks,red,mark options={fill=white},mark size=1.5pt] table[x index=9,y index=14,header=true,col sep=comma,%
    restrict expr to domain={\thisrow{suc_two}}{1:2},%
    restrict expr to domain={\thisrow{suc_one}}{1:2},%
    ] {\myfile};
    \typeout{-----------------}
    \typeout{PROCESSING: \myfile}
    \addplot [no markers, red,opacity=0.5,line width=1pt, line cap=round, dash pattern=on 0pt off 2\pgflinewidth,forget plot] gnuplot [raw gnuplot] { 
            f(x) = a*(x**b); 
            a=1; b=3; 
            fit f(x) '\myfile' u 10:15 via a,b; 
            plot [x=0:2000] f(x); 
    };
    \def\myfile{../../experiments/data/result_way_smallworld_incr_imp.csv}
    \addplot [mark=o,only marks, blue,mark options={fill=white},mark size=1.5pt] table[x index=9,y index=14,header=true,col sep=comma,%
    restrict expr to domain={\thisrow{suc_two}}{1:2},%
    restrict expr to domain={\thisrow{suc_one}}{1:2},%
    ] {\myfile};
    \typeout{-----------------}
    \typeout{PROCESSING: \myfile}
    \addplot [no markers,blue,opacity=0.5,line width=1pt, line cap=round, dash pattern=on 0pt off 2\pgflinewidth,forget plot] gnuplot [raw gnuplot] { 
            f(x) = a*(x**b); 
            a=1; b=3; 
            fit f(x) '\myfile' u 10:15 via a,b; 
            plot [x=0:2000] f(x); 
    };
    \def\myfile{../../experiments/data/result_reach_smallworld_incr_imp.csv}
    \addplot [mark=x,only marks,darkgreen,mark options={fill=white},mark size=1.5pt] table[x index=9,y index=14,header=true,col sep=comma,%
    restrict expr to domain={\thisrow{suc_two}}{1:2},%
    restrict expr to domain={\thisrow{suc_one}}{1:2},%
    ] {\myfile};
    \typeout{-----------------}
    \typeout{PROCESSING: \myfile}
    \addplot [no markers,darkgreen,opacity=0.5,line width=1pt, line cap=round, dash pattern=on 0pt off 2\pgflinewidth,forget plot] gnuplot [raw gnuplot] { 
            f(x) = a*(x**b); 
            a=1; b=3; 
            fit f(x) '\myfile' u 10:15 via a,b; 
            plot [x=0:2000] f(x); 
    };
\node [below left] at (axis cs:  \pgfkeysvalueof{/pgfplots/xmax},  \pgfkeysvalueof{/pgfplots/ymax}) {{\bf(h)}};
  \end{axis}
  \begin{axis}[%
    at={($(result_way3_smallworld_incr_imp.south)-(0,\vpgfspace)$)},anchor=north,
    xlabel style={align=center,yshift=1.0ex}, 
    ylabel style={align=center},
    ylabel={Runtime (s)},
    title style={yshift=-1.5ex},
    legend pos = north west,%
    scaled x ticks=false,
    xticklabels={X,0,10k,20k,30k,40k,50k,(rules)},
    yticklabels={x,0,500,1k},
    width=3.1in, height=1.2in,
    ymax=1000,xmax=60000,
    legend style={fill=none, font=\tiny},
    legend cell align=left,
        ]
    \def\myfile{../../experiments/data/result_way3_smallworld_incr_rule.csv}
    \addplot [mark=triangle,only marks, red,mark options={fill=white},mark size=1.5pt] table[x index=9,y index=14,header=true,col sep=comma,%
    restrict expr to domain={\thisrow{suc_two}}{1:2},%
    restrict expr to domain={\thisrow{suc_one}}{1:2},%
    ] {\myfile};
    \typeout{-----------------}
    \typeout{PROCESSING: \myfile}
    \addplot [no markers, red,opacity=0.5,line width=1pt, line cap=round, dash pattern=on 0pt off 2\pgflinewidth,forget plot] gnuplot [raw gnuplot] { 
            f(x) = a*(x**b); 
            a=1; b=3; 
            fit f(x) '\myfile' u 10:15 via a,b; 
            plot [x=0:60000] f(x); 
    };
    \def\myfile{../../experiments/data/result_way_smallworld_incr_rule.csv}
    \addplot [mark=o,only marks, blue,mark options={fill=white},mark size=1.5pt] table[x index=9,y index=14,header=true,col sep=comma,%
    restrict expr to domain={\thisrow{suc_two}}{1:2},%
    restrict expr to domain={\thisrow{suc_one}}{1:2},%
    ] {\myfile};
    \typeout{-----------------}
    \typeout{PROCESSING: \myfile}
    \addplot [no markers, blue,opacity=0.5,line width=1pt, line cap=round, dash pattern=on 0pt off 2\pgflinewidth,forget plot] gnuplot [raw gnuplot] { 
            f(x) = a*exp(x*b); 
            a=4; b=0.0001; 
            fit f(x) '\myfile' u 10:15 via a,b; 
            plot [x=0:60000] f(x); 
    };
    \def\myfile{../../experiments/data/result_reach_smallworld_incr_rule.csv}
    \addplot [mark=x,only marks, darkgreen,mark options={fill=white}] table[x index=9,y index=14,header=true,col sep=comma,%
    restrict expr to domain={\thisrow{suc_two}}{1:2},%
    restrict expr to domain={\thisrow{suc_one}}{1:2},%
    ] {\myfile};
    \typeout{-----------------}
    \typeout{PROCESSING: \myfile}
    \addplot [no markers, darkgreen,opacity=0.5,line width=1pt, line cap=round, dash pattern=on 0pt off 2\pgflinewidth,forget plot] gnuplot [raw gnuplot] { 
            f(x) = a*exp(x*b); 
            a=4; b=0.0001; 
            fit f(x) '\myfile' u 10:15 via a,b; 
            plot [x=0:60000] f(x); 
    };
\node [below left] at (axis cs:  \pgfkeysvalueof{/pgfplots/xmax},  \pgfkeysvalueof{/pgfplots/ymax}) {{\bf(i)}};
  \end{axis}
\end{tikzpicture}}
\spacehack{-1.0em}
\vspace{-1.0em}
\caption{%
{\bf(g)} Scalability of Incremental on Small-World topologies of increasing size;
{\bf(h)} Scalability when no correct switch-granularity update exists (i.e. algorithm reports ``impossible"), and
{\bf(i)} Scalability of fine-grained (rule-granularity) approach for solving switch-impossible examples in (h). 
}
\spacehack{-.75em}
\label{fig_scalability}
\vspace{-0.35em}
\end{figure}

\paragraph*{Incremental vs. NetPlumber.} 
We also measured the performance of Incremental versus the network
property checker NetPlumber (Figure \ref{fig_netplumber}(d-f)). 
Note that NetPlumber uses rule-granularity 
for updates, so we enabled this mode in our tool
for these experiments. For the three datasets, our checker is faster
on all experiments, with mean speedups of ($6.41$x, $4.90$x,
$17.19$x). NetPlumber does not report counterexamples, putting
it at a disadvantage in this end-to-end comparison, so we also
measured total Incremental versus NetPlumber runtime on the same set
of model-checking questions posed by Incremental for the Small-World
example. Our tool is still faster on all instances,
with a mean speedup of $2.74$x.

\paragraph*{Scalability.}
To quantify our tool's scalability, we constructed Small World
topologies with up to 1500 switches, and ran experiments with large
diamond updates---the largest has $1015$
switches updating. The results appear in
Figure \ref{fig_scalability}(g).
The maximum synthesis times for the three properties were $129.04$s, $30.11$s, and $0.85$s, which shows that our tool
scales to problems of realistic size.

\paragraph*{Infeasible Updates.}
We also considered examples for which there is no 
switch-granular update.  Figure~\ref{fig_scalability}(h) shows the
results of experiments where we generated a second diamond atop the
first one, requiring it to route traffic in the opposite
direction. Using switch-granularity, the inputs are reported as
unsolvable in maximum time $153.48$s, $33.48$s, and $0.69$s. Using
rule-granularity, these inputs are solved successfully for up to 1000
switches with maximum times of $776.13$s, $512.84$s, and $82.00$s (see
Figure~\ref{fig_scalability}(i)).

\paragraph*{Waits.}
We also separately measured the time needed to
run the wait-removal heuristic
for the Figure~\ref{fig_scalability} experiments.
For (g), the maximum
wait-removal runtime was $0.89$s, resulting in $2$ needed waits for
each instance. For (i), the maximum wait-removal runtime was
$103.87$s, resulting in about $2.6$ waits on average (with a maximum of
$4$). For the largest problems in (g) and (i), this corresponds to
removal of $1397/1399$ and $55823/55826$ waits (about 99.9\%).

\optnewpage
\section{Related Work \budget{1}}
\label{related}

This paper extends preliminary work 
reported in a workshop
paper~\cite{noyestoward}. We present a more precise and
realistic network model, and replace
expensive calls to an external model checker with calls to a
new built-in {\em incremental} network model checker.
We extend the DFS search
procedure with optimizations and heuristics that improve performance
dramatically. Finally, we evaluate our tool on 
a comprehensive set of benchmarks with real-world topologies. %

\paragraph*{Synthesis of concurrent programs.}
There is much previous work on synthesis for concurrent
programs~%
\cite{vechev2010abstraction,solar2008sketching,hawkins2012concurrent}.
In particular, work by Solar-Lezama et al.~\cite{solar2008sketching}
and Vechev et al.~\cite{vechev2010abstraction} synthesizes sequences of
instructions. However, traditional
synthesis and synthesis for networking are quite
different. First, traditional 
synthesis is a game against the environment which (in the concurrent
programming case) provides inputs and schedules threads; in contrast,
our synthesis problem involves reachability on the space of
configurations. Second, our space of configurations is very
rich, meaning that checking configurations is itself a model checking problem.

\paragraph*{Network updates.} There are many protocol- and
property-specific algorithms for implementing network updates,
e.g. avoiding packet/bandwidth loss during
planned maintenance to
BGP~\cite{francois2007avoidingDisrupt,raza2011graceful}. 
Other work avoids routing loops and blackholes during IGP
migration \cite{vanbever2011seamless}. 
Work on network updates in SDN proposed
the notion of {\it consistent updates} and several
implementation mechanisms, including two-phase
updates~\cite{reitblatt2012abstractions}.
Other work explores propagating updates incrementally, 
reducing the space overhead on switches~\cite{katta2013incremental}.
As mentioned in Section~\ref{overview}, recent work proposes ordering
updates for specific 
properties~\cite{dionysus14}, whereas we can handle
combinations and variants of these properties. 
Furthermore, SWAN and
zUpdate add support for bandwidth
guarantees~\cite{hong2012swan,liu2013zupdate}.   
Zhou et al.~\cite{ZJCCG15} consider customizable trace properties, and
propose a dynamic algorithm to find order updates. This solution can
take into account unpredictable delays caused by switch
updates. However, it may not always find a solution, even if one
exists. In contrast, we obtain a completeness guarantee for our static
algorithm. Ludwig et al.~\cite{LRFS14} consider ordering updates for 
waypointing properties. 

\paragraph*{Model checking.} Model checking has been used for network
verification~\cite{al2010flowchecker,mai2011debugging,kazemian2012header,khurshid2012veriflow,MTW14}. 
The closest to our work is the incremental checker NetPlumber~\cite{Kazemian}.
Surface-level differences include the specification languages (LTL
vs. regular expressions), and NetPlumber's lack of counterexample
output.
The main difference is incrementality: 
Netplumber restricts checking to ``probe nodes,"
keeping track of ``header-space" reachability information for those
nodes, and then performing property queries based on this.
In contrast, we look at the {\em property}, keeping
track of {\em portions of the property} holding at each node,
which keeps incremental rechecking times low.
The empirical comparison (Section~\ref{evaluation}) showed better
performance of our tool as a back-end for
synthesis. %

Incremental model checking has been studied
previously, with
\cite{SS94} presenting the first incremental
model checking algorithm, for alternation-free $\mu$-calculus. 
We consider LTL properties and specialize our algorithm to exploit the
no-forwarding-loops assumption. 
The paper \cite{CIMMN11} introduced an incremental algorithm, but
it is specific to the type of partial results produced by
IC3~\cite{B11}.  

\section{Conclusion}
\label{sec:conclusion}

We present a practical tool for automatically synthesizing
correct network update sequences from formal specifications. We discuss 
an efficient 
incremental model checker that performs orders of magnitude better
than state-of-the-art monolithic tools. Experiments on real-world
topologies demonstrate the effectiveness of our approach for
synthesis. In future work, we plan to explore both extensions to deal
with network failures and bandwidth constraints,
and deeper foundations of techniques for network updates.  

\ifanon\else

\acks

The authors would like to thank the PLDI reviewers and AEC members for their
insightful feedback on the paper and artifact, as well as Xin Jin, Dexter Kozen, 
Mark Reitblatt, and Jennifer Rexford for helpful 
comments. Andrew Noyes and Todd Warszawski contributed a number of
early ideas through an undergraduate research project. 
This work was supported by the NSF under awards CCF-1421752,
CCF-1422046, CCF-1253165, CNS-1413972, CCF-1444781, and CNS-1111698;
the ONR under Award N00014-12-1-0757; and gifts from Fujitsu Labs and
Intel.

\fi

\bibliographystyle{abbrvnat}
\bibliography{paper}


\optappendix {
\appendix
\section{Network Model Auxiliary Definitions}
\label{appnet}

We
first define what it means for a table to be
active, i.e. the controller contains an update that
will eventually produce that table.

\begin{definition}[Active Forwarding Table] 
Let $N$ be a network. The forwarding table $\tbl$ is \emph{active 
  in the epoch $\Epoch$} for the switch $\sw$ if
\begin{compactenum}
\item $\Epoch = 0$ and $\tbl$ is the initial table of $\sw$ in $N$, or
\item $\Epoch > 0$ and either (a) if there exists a command $(\sw',\tbl') \in C.\cmds$ such
that $\sw = \sw'$ and the number of 
$\wait$ commands preceding $(\sw,\tbl)$ in $C.\cmds$ is
$\Epoch$, then $\tbl=\tbl'$, or (b) if there does not exist such a
command, then $\tbl$ is the table active for the switch $\sw$ in epoch
$\Epoch-1$.  
\end{compactenum}
\end{definition}

\noindent Next we define what it means for an observation $o'$ to
succeed $o$.

\begin{definition}[Successor Observation] 
\label{def:succ}
Let $N$ be a network and let $o = (\sw,\pt,\pkt)$ and $o'
=(\sw',\pt',\pkt')$ be observations.  The observation $o'$ is a
\emph{successor of $o$ in $\Epoch$}, written $o
\overset{\Epoch}{\sqsubseteq} o'$, if either: 
\begin{compactitem}
\item there exists a switch $S_i$ and link $L_j$ such that $S_i.\sw
=  \sw$ and $S_i.\tbl$ is active in $\Epoch$ and $L_j.\loc = 
  (\sw,\pt_j)$ and $L_j.\loc' = (\sw',\pt')$ and $(\pt_j,\pkt') \in
  \denot{S_i.\tbl}(\pt,\pkt)$, or
\item there exists a switch $S_i$, a link $L_j$, and a host $\host$
  such that $S_i.\sw = \sw$ and $S_i.\tbl$ is active in $\Epoch$
  and $L_j.\loc = (\sw,\pt')$ and $L_j.\loc' = \host$ and
  $(\pt',\pkt') \in \denot{S_i.\tbl}(\pt,\pkt)$.
\end{compactitem}
\end{definition}

\noindent Intuitively $o \overset{\Epoch}{\sqsubseteq} o'$ if the
packet in $o$ could have directly produced the packet in $o'$ in
$\Epoch$ by being processed on some switch. The two cases correspond
to an internal and egress processing steps.

\begin{definition}[Unconstrained Single-Packet Trace]
\label{defsingle_all}
Let $N$ be a network. The sequence $(o_1 \cdots o_l)$ is a \emph{unconstrained single-packet
  trace of $N$} if \( N \steps{o_1'} \dots \steps{o_k'} N_k \) such
that $(o_1 \cdots o_l)$ is a subsequence of $(o_1' \cdots o_k')$ for which
\begin{compactitem}
\item every observation is a successor of the preceding observation in
monotonically increasing epochs, and
\item if $o_1=o_j'=(\sw,\pt,\pkt)$, i.e.
\( N \steps{o_1'} \dots \steps{o_j'=o_1} N_j \steps{o_{j+1}'} \dots \steps{o_k'} N_k \),
then no $o_i' \in \{o_1', \cdots, o_{j-1}'\}$ precedes $o_1$, and
\item the $o_l$ transition is an {\sc Out}
terminating at a host.
\end{compactitem}
\end{definition}

\noindent Unconstrained single-packet traces %
are not required to begin at a host.
We write $\TracesAll{N}$ for the set of unconstrained single-packet traces generated by $N$,
and note that $\Traces{N} \subseteq \TracesAll{N}$.

\begin{definition}[Network Kripke Structure] 
\label{def:netkripke}
Let $N$ be a static network. We define a Kripke structure $\Kripke{N} = (Q,Q_0,\delta,\lambda)$ as follows. 
The set of states $Q$ comprises tuples of the form $(\sw,\pt,T_k)$. 
The set $Q_0$ contains states $(\sw,\pt,T_k)$ where $\sw$ and $\pt$ are adjacent to
an ingress link---i.e., there exists a link $L_j$ and host $\host$
such that $L_j.\loc = \host$ and $L_j.\loc' = (\sw,\pt)$. 
Transition relation $\delta$ contains all pairs of states
$(\sw,\pt,T_k)$ and $(\sw',\pt',T_k')$ where there exists a switch $S$
and a link $L$ such that $S.\sw = \sw$ and either:
\begin{compactitem}
\item there exists a link $L_j$ and packets $\pkt \in T_k$ and
  $\pkt'\in T_k'$ such that $L.\loc' = (\sw, pt)$ and $L_j.\loc =
  (\sw,\pt_j)$ and $L_j.\loc' = (\sw',\pt')$ and $(\pkt',\pt_j) \in
  \denot{S.\tbl}(\pkt,\pt)$.
\item there exists a link $L_j$, a host $\host$, and packets $\pkt \in
  T_k$ and $\pkt' \in T_k'$ such that $L.\loc' = (\sw, pt)$ and
  $L_j.\loc = (\sw,\pt')$ and $L_j.\loc' = \host$ and $(\pkt',\pt')
  \in \denot{S.\tbl}(\pkt,\pt)$.
\item $(\sw,\pt,T_k) = (\sw',\pt',T_k')$ and there exists a packet
  $\pkt \in T_k$ such that $L.\loc' = (\sw, pt)$ and $\denot{S.\tbl}
  (\pkt,\pt) = \{\}$.
\item $(\sw,\pt,T_k) = (\sw',\pt',T_k')$ and there exists a link $L_j$
  and host $\host$ such that $L_j.\loc = (\sw,\pt)$ and $L_j.\loc' =
  \host$.
\end{compactitem}
Finally, the labeling function $\lambda$ maps each state
$(\sw,\pt,T_k)$ to $T_k$, which captures the set of all possible
header values of packets located at switch $\sw$ and port $\pt$.
\end{definition}
\noindent The four cases of the $\delta$ relation correspond to
forwarding packets to an internal link, forwarding packets out an
egress, dropping packets on a switch, or
reaching an egress (inducing a self-loop).

We can relate the observations generated by a network $N$ and the
traces of the Kripke structure generated from it.
\begin{definition}[Trace Relation]
\label{def:traceequiv}
Let $N$ be a static network and $K$ a Kripke structure. Let $\lesssim$
be a relation on observations of $N$ and states of $K$ defined by
$(\sw,\pt,\pkt)~\lesssim~(\sw,\pt,T_k)$ if and only if $\pkt \in
T_k$. Lift $\lesssim$ to a relation on (finite) sequences of
observations and (infinite) traces by repeating the final observation
and requiring $\lesssim$ to hold pointwise: $o_1 \cdots
o_k~\lesssim~t$ if and only if $o_i~\lesssim~t_i$ for $i$ from $1$ to
$k$ and $o_k~\lesssim~t_j$ for all $j > k$.
\end{definition}

\begin{lemma}[Traces of a Stable Network]
Let $N$ be a stable network. Then for each trace $t \in \TracesAll{N}$, there exists a trace $t' \in \Traces{N}$
such that $t$ is a {\em suffix} of $t'$.
\end{lemma}

\begin{lemma}[Trace-Equivalence]
\label{lem:trace:eq}
Let $N_1, N_n$ be static networks where $N_1 \rightarrow \cdots \rightarrow N_n$
and no transition is an update command. For a single-packet trace $t$, we have $t \in \Trace(N_1) \iff t \in \Trace(N_n)$.
\end{lemma}

\begin{lemma}[Induced Sequence of Networks]
\label{lem:induce:seq}
Let $N_1$ be a static network, and let $N_1'$ be the network obtained by emptying all packets from $N_1$.
Let $\cmds$ be a sequence of commands, and let $c_1 \cdots c_{n-1}$ be the subsequence of update commands.
Construct the sequence $N_1' \rightarrow \cdots \rightarrow N_n'$ of empty networks by executing the update
commands in order. Now, given any sequence $N_1 \rightarrow \cdots \rightarrow N_n$ induced by $\cmds$,
we have $N_i \simeq N_i'$ for all $i$.
\end{lemma}
 
\noindent In other words, any induced sequence of static networks is pointwise trace-equivalent to the
unique sequence of {\em network configurations} generated by running the update commands in order.

\section{Synthesis Algorithm Correctness Proofs}
\label{app:dfs}

{
\renewcommand{\thelemma}{\ref{lem:kripke:sound}}
\begin{lemma}[Network Kripke Structure Soundness]
Let $N$ be a static network and $K = \Kripke{N}$ a network Kripke
structure. For every single-packet trace $t$ in $\Traces{N}$ there
exists a trace $t'$ of $K$ from a start state such that $t \lesssim
t'$, and vice versa.
\end{lemma}
\addtocounter{lemma}{-1}
}
\begin{proof}
We proceed by induction over $k$, the length of the (finite prefix of the) trace.
The base case $k=1$ is easy to see, since the lone observation in $t$ must
be on an ingress link, meaning the corresponding state in $K$ will be an
initial state with a self-loop (case 3 of Definition \ref{def:netkripke}),
and these are equivalent via Definition \ref{def:traceequiv}.

For the inductive step ($k > 1$), we wish to show both directions of subtrace relation $\lesssim$
to conclude equivalence. First, let $t=o_1,\cdots,o_{k+1}$ be a single-packet trace of length $k+1$
in $\Traces{N}$, and we must show that  
$\exists t' \in \Kripke{N}$ such that $t \lesssim t'$. Let $t^k$ be the prefix of $t$
having length $k$. By our induction hypothesis, there exists $t'^k = s_1,\cdots,s_{k-1},s_k,s_k,\cdots \in \Kripke{N}$
such that $t^k \lesssim t'^k$. We have the successor relation $o_{k} \sqsubseteq o_{k+1}$,
so Definition \ref{def:succ} and \ref{def:netkripke} tells us that we have a transition $s_k \rightarrow s'$
for some $s' \in K$. We see that this $s'$ is exactly what we need to construct $t'=s_1,\cdots,s_k,s',s',\cdots$
which satisfies the relation $t \lesssim t'$.

Now, let $t' = s_1,\cdots,s_k,s_{k+1},s_{k+1},\cdots$ be a trace in $\Kripke{N}$ for which the finite prefix has length $k+1$.
We must show that $\exists t \in \Traces{N}$ such that $t \lesssim t'$.
Let $t'^k = s_1,\cdots,s_{k-1},s_k,s_k,\cdots$, and by our induction hypothesis,
and there exists $t^k = o_1,\cdots,o_k$ such that $t^k \lesssim t'^k$.
Consider transition $s_k \rightarrow s_{k+1}$. If $s_k = s_{k+1}$, then $t' = t'^k$, so we can 
let $t = t^k$, and conclude that $t \lesssim t'$. Otherwise, if $s_k \not= s_{k+1}$, then we have one of
the first two cases in Definition \ref{def:netkripke}, which correspond to the cases in Definition \ref{def:succ},
allowing us to construct an $o_{k+1}$ such that $o_{k} \sqsubseteq o_{k+1}$. We let $t = o_1,\cdots,o_k,o_{k+1}$,
and conclude that $t \lesssim t'$.
\end{proof}


\noindent We want to develop a lemma showing that the correctness of careful command sequences can be
reduced to the correctness of each induced $N_i$, so we start with the following
auxiliary lemma:

\begin{lemma}[Traces of a Careful Network]
\label{lem:careful:trace}
Let $N$ be a stable network with $C.\cmds$ careful, and consider a sequence of
static networks induced by $C.\cmds$. For every trace $t \in
\Traces{N}$ there exists a stable static network $N_i$ in the sequence
s.t. $t \in \Traces{N_i}$.
\end{lemma} 

\begin{proof}
I. First, we show that at most one update transition can be involved in the
trace. In other words, if \( N \steps{o_1'} \dots \steps{o_k'} N_k \)
where $t = o_1 \cdots o_n$ is a subsequence of $o_1' \cdots o_k'$, and if
$f : \mathbb{N} \rightarrow \mathbb{N}$ is a bijection between $o_i$
indices and $o_i'$ indices, then at most one of
the transitions $o_{f(1)}',\cdots,o_{f(n)}'$ is an {\sc Update} transition.

Assume to the contrary that there are more than one such transitions,
and consider two of them, $o_i', o_j'$ where $i,j\in\{f(1),\cdots,f(n)\}$,
assuming without loss of generality that $i < j$.
Now, since the sequence $C.cmds$ is careful, we must have
both an {\sc Incr} and {\sc Flush} transition between $o_i'$ and $o_j'$.
This means that the second update $o_j'$ cannot happen while the
trace's packet is still in the network, i.e. $j > f(n)$,
and we have reached a contradiction.

II. Now, if there are zero update transitions, we are done, since the trace
is contained in the first static $N$. If there is one update
transition $N_{k+1} = N_k[\sw \gets \tbl]$, and this update occurs
before the packet reaches $\sw$ in the trace, then the trace is fully
contained in $N_{k+1}$. Otherwise, the trace is fully contained in $N_k$.
\end{proof}

\balance %

{
\renewcommand{\thelemma}{\ref{lem:careful:corr}}
\begin{lemma}[Careful Correctness]
Let $N$ be a stable network with $C.\cmds$ careful and let $\varphi$ be an LTL
formula. If $\cmds$ is careful
and $N_i
\models \phi$ for each static network in any sequence induced by
$\cmds$,
then $\cmds$ is correct with respect to $\varphi$.
\end{lemma} 
\addtocounter{lemma}{-1}
}

\begin{proof}
Consider a trace $t \in \Trace(N)$. From Lemma \ref{lem:careful:trace},
we have $t \in \Trace(N_i)$ for some $N_i$ in the induced sequence.
Thus $t \models \varphi$, since our hypothesis tells us that
$N_i \models \varphi$. Since this is true for an arbitrary trace,
we have shown that $\Trace(N) \models \varphi$, i.e. $N \models \varphi$,
meaning that $\cmds$ is correct with respect to $\varphi$.
\end{proof}

{
\renewcommand{\thetheorem}{\ref{thm:sound}}
\begin{theorem}[Soundness]
Given initial network $\NetPolicy_i$, final configuration
$\NetPolicy_f$, and LTL formula $\varphi$, if \textsc{OrderUpdate}
returns a command sequence $\cmds$, then
$\Step{\NetPolicy_i}{\cmds}{\NetPolicy'}$ s.t. $\NetPolicy' \simeq \NetPolicy_f$, and $\cmds$ is correct
with respect to $\varphi$ and $\NetPolicy_i$.
\end{theorem}
\addtocounter{theorem}{-1}
}

\begin{proof}
It is easy to show that if \textsc{OrderUpdate} returns $\cmds$, then
$\Step{\NetPolicy_i}{\cmds}{\NetPolicy'}$ where $\NetPolicy' \simeq \NetPolicy_f$. 
Each update in the returned sequence changes a switch configuration of one switch $\Switch$ to
the configuration $\NetPolicy_f(\Switch)$, and the algorithm terminates when
all (and only) switches $\Switch$ such that $\NetPolicy_i(\Switch)
\neq \NetPolicy_f(\Switch)$ have been updated.

Observe that if \textsc{OrderUpdate} returns $\cmds$, the sequence can be
made careful by choosing an adequate time delay between each update command,
and for all $j \in \{0,\cdots,n\}$, $\NetPolicy_j \models \varphi$. 
This is ensured by the call to a model checker (Line~\ref{line:modelcheck}).
We use Lemma~\ref{lem:careful:corr} to conclude that $\cmds$ is
correct with respect to $\varphi$ and $\NetPolicy_i$.
\end{proof}

To show that \textsc{OrderUpdate} is complete with respect to simple
and careful command sequences, we observe that \textsc{OrderUpdate}
searches through all simple and careful sequences.

{
\renewcommand{\thetheorem}{\ref{prop:complete}}
\begin{theorem}[Completeness]
Given initial network $\NetPolicy_i$, final configuration
$\NetPolicy_f$, and specification $\varphi$, if there exists a
simple, careful sequence $\cmds$ with
$\Step{\NetPolicy_i}{\cmds}{\NetPolicy'}$ s.t. $\NetPolicy' \simeq \NetPolicy_f$, then
\textsc{OrderUpdate} returns one such sequence.
\end{theorem}
\addtocounter{theorem}{-1}
}


\section{Incremental Checking Correctness Proofs}
\label{app:incrmodcheck}

{
\renewcommand{\thelemma}{\ref{prop:label1state}}
\begin{lemma}
First, $\HoldsSink(q,M) \Leftrightarrow \exists t \in \traces(q):
t \models M$ for sink states $q$. Second, if $\labelGraph_K$ is a
correct labeling with respect to $\varphi$ and $\succs_K(q)$,
then \( \Holds_K(q,M,\labelGraph_K) \Longleftrightarrow \exists
t \in \traces_K(q): t \models M \). 
\end{lemma}
\addtocounter{lemma}{-1}
}

\begin{proof}%
First, for sink states, observe that there is a unique trace $t$ in
$\traces(q)$, as $q$ is a 
sink state.  We first prove that $t \models \varphi$ iff
$\HoldsS(q,\varphi)$. We prove this by induction on the structure of
the LTL formula.  Then we observe that there is a unique
maximally-consistent set $M$ such that $t \models M$.  This is the set
$\{ \psi \mid t \models \psi \wedge \psi \in \ecl(\varphi)\}$.  We
then use the definition of $\HoldsSink(q, M)$ for sink states to
conclude the proof.

Now consider non-sink states: we first prove soundness, i.e., if
$\Holds_K(q,M,\labelGraph_K)$, 
then there exists $t \in traces(q)$ such that $t \models M$. 
We have $\Holds_K(q,M,\labelGraph_K)$ iff $(\lambda(q) = (\AP \cap
M))$ and there exists $q' \in \succs_K(M)$, and 
$M' \in \labelGraph_K(q')$ such that $\follows(M,M')$. By assumption
of the theorem, we have that if $M' \in 
\labelGraph_K(q')$, then there exists a trace $t'$ in $\traces(q')$
such that $t' \models M'$. 
Consider a trace $t$ such that $t_0 = q$ and
$t^1 = t'$. For each $\psi \in M$, we can prove that $t \models \psi$
as follows.  
The base case of the proof by induction is implied by the fact that
$q \models (\AP \cap M)$. The inductive cases are proven using the 
definitions of maximally-consistent set and the function $\follows$.  
We now prove completeness, i.e., that if there exists a trace $t$ in
$\traces_K(q)$ such that $t \models M$, then
$\Holds_K(q,M,\labelGraph_K)$ is true. Let $t$ be the trace $q q_1 q_2
\ldots$. 
It is easy to see that if $M$ is a maximally-consistent set,
and $t \models M$, then $M = \{ \psi \mid \psi \in \ecl(\varphi)
\wedge t \models \psi\}$.  
Let us consider the set of formulas $S = \{\psi \mid \psi \in
\ecl(\varphi) \wedge t^1 \models \psi\}$. 
Observe that $S$ is a maximally-consistent set.  
By assumption of the theorem, we have that $S$ is in
$\labelGraph_K(q_1)$. 
It is easy to verify that $\follows(M,S)$.
\end{proof}

{
\renewcommand{\thetheorem}{\ref{thm:relabel}}
\begin{theorem}
Let $V \subseteq Q$ be a set of vertices and $\labelGraph_K$ a correct
labeling with respect to $\varphi$ and $Q \setminus \ancestors_K(V)$.
Then $\relabel_K(\varphi,\labelGraph_K,V)$ is a correct labeling w.r.t.
$\varphi$ and $Q$.
\end{theorem}
\addtocounter{theorem}{-1}
}

\begin{proof}%
We first note that only ancestors of nodes in $V$ are re-labeled---%
all the other nodes are correctly labeled by assumption on
$\labelGraph$. 
We say that a node $q$ is at level $k$ w.r.t. a set of vertices
$T$ iff the longest simple path from $q$ to a node in $T$ is $k$.    
Let $H_k$ be the set of nodes at level $k$ from $V$. 
We prove by induction on $k$ that at $k$-th iteration,
we have a correct labeling of $K$ w.r.t. $\varphi$ and
$(S \setminus \ancestors_K(V)) \cup H_k$, where $S$ is the set of
states of $K$. 
We can prove the inductive claim using Lemma~\ref{prop:label1state}. 
\end{proof}

{
\renewcommand{\thecorollary}{\ref{cor:modcheck}}
\begin{corollary}
First, $\modelCheck_K(\varphi)=true \iff K \models \varphi$. Second,
for $(K,K',U)$ and $\labelGraph_K$ as above, we have 
$\incrModelCheck(K,\varphi,U,\labelGraph_K)=\mathit{true} \iff
K \models \varphi$.   
\end{corollary}
\addtocounter{corollary}{-1}
}

\begin{proof}%
Using Theorem~\ref{thm:relabel}, and the fact that the set
$\ancestors_K(S_f)$ is the set $S$ of all states $K$, we obtain that
$\labelGraph_K = 
\relabel_K(\varphi,\labelGraph_K^0,S_f)$ is a correct labeling of $K$ with
respect to $\varphi$ and $S$. 
In particular, for all initial states $q_0$, we have that for all $M \subset
\ecl(\varphi)$, $m \in \labelGraph_K(q_0)$ iff there exists a trace $t \in
\traces_K(q_0)$ such that $t \models M$.  
We now use the definition
of $\checkInitStates$ to show that if $\checkInitStates$ returns
true, then there is no initial state $q_0$ such that there exists $M
\in \labelGraph_K(q_0)$ such that $\neg\varphi \in M$. 
Thus for all initial
states $q_0$, for all traces $t$ in
$\traces(t_0)$, we have that $t \models \varphi$. 

The proof for incremental model checking is similar. 
\end{proof}

\newcommand{\exprtable}[3]{
\begin{table}[h]
\centering
\pgfplotstabletypeset[%
font=\scriptsize,
fixed,precision=2,fixed zerofill,
columns/filename/.style={column name=File Name,verb string type},%
columns/rules/.style={column name=Rl. (\#),precision=0},%
columns/nt_rules/.style={column name=Nontriv. Rl. (\#),precision=0},%
columns/switches/.style={column name=Sw. (\#),precision=0},%
columns/nontriv/.style={column name=Nontriv. Sw. (\#),precision=0},%
columns/suc_one/.style={column name=Answer,precision=0},%
columns/tim_grph/.style={column name=Graph Gen. (s)},%
columns/time_one/.style={column name=Synth. Tot. (s)},%
columns/chk_time/.style={column name=Model Check. (s)},%
columns/wrm_time/.style={column name=Wt. Rem. (s)},%
columns/num_wts/.style={column name=Wt. (\#),precision=0},%
every head row/.style={after row=\midrule},
col sep=comma,
     columns={filename,suc_one,switches,nontriv,rules,nt_rules,tim_grph,chk_time,time_one,wrm_time,num_wts},
    ]{../../experiments/data/#1.csv}
\caption{#2}
\end{table}
}

}


\end{document}